%% file: main.tex
\documentclass[sigconf]{acmart}
\input{macros.tex} 
\input{defs.tex} 

\copyrightyear{2018} 
\acmYear{2018} 
\setcopyright{acmlicensed}
\acmConference[CIKM '18]{The 27th ACM International Conference on Information and Knowledge Management}{October 22--26, 2018}{Torino, Italy}
\acmBooktitle{The 27th ACM International Conference on Information and Knowledge Management (CIKM '18), October 22--26, 2018, Torino, Italy}
\acmPrice{15.00}
\acmDOI{10.1145/3269206.3271756}
\acmISBN{978-1-4503-6014-2/18/10}

\begin{document}
\title[Probabilistic Causal Analysis of Social Influence]{Probabilistic Causal Analysis of Social Influence}

\author{Francesco Bonchi}
\affiliation{\small ISI Foundation, Italy}
\email{francesco.bonchi@isi.it}

\author{Francesco Gullo}
\affiliation{\small UniCredit, R\&D Dept., Italy}
\email{gullof@acm.org}

\author{Bud Mishra}
\affiliation{\small New York University, USA}
\email{mishra@nyu.edu}

\author{Daniele Ramazzotti}
\affiliation{\small Stanford University, CA, USA}
\email{daniele.ramazzotti@stanford.edu}

\renewcommand{\shortauthors}{F. Bonchi~\emph{et al.}}


\begin{abstract}
Mastering the dynamics of social influence requires separating, in a database of information propagation traces, the genuine causal processes from temporal correlation, i.e., homophily and other spurious causes. However, most studies to characterize social influence, and, in general, most data-science analyses focus on correlations, statistical independence, or conditional independence.
Only recently, there has been a resurgence of interest in ``causal data science,'' e.g., grounded on causality theories. In this paper we adopt a principled causal approach to the analysis of social influence from information-propagation data, rooted in the theory of probabilistic causation.

Our approach consists of two phases. In the first one, in order to avoid the pitfalls of misinterpreting causation when the data spans a mixture of several subtypes (``Simpson's paradox''), we partition the set of propagation traces into groups, in such a way that each group is as less contradictory as possible in terms of the hierarchical structure of information propagation. To achieve this goal, we borrow the notion of ``\emph{\agony}''~\cite{GupteSLMI} and define the \clusteringprob problem, which we prove being hard, and for which we develop two efficient algorithms with approximation guarantees. In the second phase, for each group from the first phase, we apply a constrained MLE approach to ultimately learn a minimal causal topology.
Experiments on synthetic data show that our method is able to retrieve the genuine causal arcs w.r.t. a ground-truth generative model. Experiments on real data show that, by focusing only on the extracted causal structures instead of the whole social graph, the effectiveness of predicting influence spread is significantly improved.
\end{abstract}
\maketitle
\sloppy

\vspace{3mm}
\section{Introduction}
\label{sec:intro}
\input{intro.tex}

\section{Related Work}
\label{sec:related}
\input{rw.tex}

\section{Preliminaries}
\label{sec:framework}
\input{framework.tex}

\section{Problem statement}
\label{sec:agony_bounded_clustering_problem}
\input{agony_clustering.tex}
\section{Algorithms}
\label{sec:algorithms}
\input{algorithms.tex}

\section{Experiments}
\label{sec:experiments}
\input{experiments_new.tex}

\section{Conclusions}
\label{sec:conclusions}
\input{conclusions.tex}

\balance
\bibliographystyle{abbrv}
\bibliography{references,propagation}

\end{document}

%% file: macros.tex
\usepackage{graphicx}
\usepackage[show]{chato-notes}
\usepackage{amsmath,amssymb}
\usepackage{epsfig}
\usepackage{algorithm}
\usepackage[noend]{algorithmic}
\usepackage{url}
\usepackage{hyperref}
\usepackage{caption}
\usepackage{setspace}
\usepackage{multirow}

\hypersetup{
    bookmarks=true,         
    unicode=false,          
    pdftoolbar=true,        
    pdfmenubar=true,        
    pdffitwindow=false,     
    pdfstartview={FitH},    
    pdftitle={My title},    
    pdfauthor={Author},     
    pdfsubject={Subject},   
    pdfcreator={Creator},   
    pdfproducer={Producer}, 
    pdfnewwindow=true,      
    colorlinks=true,       
    linkcolor=black,          
    citecolor=black,        
    filecolor=black,      
    urlcolor=black           
}

\newcommand{\hide}[1]{}

\newcommand{\spara}[1]{\smallskip\noindent{\bf #1}}

\newtheorem{mydefinition}{Definition}
\newtheorem{mytheorem}{Theorem}

\newtheorem{problem}{Problem}

\newcommand{\NPhard}{$\mathbf{NP}$-hard}
\newcommand{\NPcomplete}{$\mathbf{NP}$-complete}

\newcommand{\squishlist}{
 \begin{list}{$\bullet$}
  {  \setlength{\itemsep}{0pt}
     \setlength{\parsep}{3pt}
     \setlength{\topsep}{3pt}
     \setlength{\partopsep}{0pt}
     \setlength{\leftmargin}{2em}
     \setlength{\labelwidth}{1.5em}
     \setlength{\labelsep}{0.5em}
} }
\newcommand{\squishlisttight}{
 \begin{list}{$\bullet$}
  { \setlength{\itemsep}{0pt}
    \setlength{\parsep}{0pt}
    \setlength{\topsep}{0pt}
    \setlength{\partopsep}{0pt}
    \setlength{\leftmargin}{2em}
    \setlength{\labelwidth}{1.5em}
    \setlength{\labelsep}{0.5em}
} }

\newcommand{\squishdesc}{
 \begin{list}{}
  {  \setlength{\itemsep}{0pt}
     \setlength{\parsep}{3pt}
     \setlength{\topsep}{3pt}
     \setlength{\partopsep}{0pt}
     \setlength{\leftmargin}{1em}
     \setlength{\labelwidth}{1.5em}
     \setlength{\labelsep}{0.5em}
} }

\newcommand{\squishend}{
  \end{list}
}

\newcommand{\twitter}{\textsf{Twitter}}
\newcommand{\lastfm}{\textsf{Last.fm}}
\newcommand{\flixster}{\textsf{Flixster}}

\newcommand{\us}{\textsf{PSC}}
\newcommand{\baseline}{\textsf{B}}

%% file: defs.tex
\usepackage{xspace}
\usepackage{balance}

\newcommand{\E}{\ensuremath{\mathcal{E}}}
\newcommand{\A}{\ensuremath{\mathbb{O}}}
\newcommand{\D}{\ensuremath{\mathbb{D}}}
\newcommand{\DAG}{\textsc{dag}\xspace}
\newcommand{\DAGs}{\textsc{dag}s\xspace}

\newcommand{\agony}{agony}

\newcommand{\Probab}[1]{\mathcal{P}({#1})}
\newcommand{\Pcond}[2]{\Probab{{#1}\mid{#2}}}

\renewcommand{\~}[1]{\overline{#1}}

\newcommand{\minimaltopology}{\textsc{Minimal Causal Topology}\xspace}

\newcommand{\clusteringprob}{\mbox{\textsc{Agony-bounded Partitioning}}\xspace}
\newcommand{\twostepalg}{\mbox{\textsf{Two-step-Agony-Partitioning}}\xspace}
\newcommand{\samplingSubroutine}{\mbox{\textsf{Sample-Maximal-}\DAG-\textsf{set}}\xspace}
\newcommand{\setcover}{\mbox{\textsc{Set Cover}}\xspace}
\newcommand{\proposedalg}{\mbox{\textsf{Sampling-Agony-Partitioning}}\xspace}

\newcommand{\bth}{\boldsymbol{\theta}}

%% file: intro.tex


Sophisticated empirical analyses of a variety of social phenomena have now become possible as a result of two related developments: the explosion of on-line social networks and the consequent unprecedented availability of network data.
Among the phenomena investigated, one that has attracted a lion's share of interest is the study of \emph{social influence}, i.e., the causal processes motivating the actions of a user to induce similar actions from her peers, creating a selective sweep of behavioral memes.
Mastering the dynamics of social influence (i.e., observing, understanding, and measuring it) could pave the way for many important applications, among which the most prominent is \emph{viral marketing}, i.e., the exploitation of  social influence by letting adoption of new products to hitch-hike on a viral meme sweeping through a social influence network~\cite{domingos01,kempe03,surveybonchi11}.
Moreover, patterns of influence can be taken as a proxy for trust and exploited in trust-propagation analysis~\cite{Guha04,Ziegler05,Golbeck06,Mohsen08} in large networks and in P2P systems.
Other applications include personalized recommendations~\cite{Song06,song07}, feed ranking in social networks~\cite{NectaRSS,memerank}, and the analysis of  information propagation in social media~\cite{weng2010, bakshy2011,CastilloMP11,RomeroMK11}.

Considerable attention has been devoted to the problem of estimating the strength of influence between two users in a social network~\cite{Saito08,amit2010,KutzkovBBG13}, mainly by recording how many times information successfully propagates from one vertex to the other. However, by social influence it is understood to mean a \emph{genuine causal process}, i.e., a force that an individual exerts on her peers to an extent to introduce change in opinions or behaviors. Thus, to properly deal with influence-induced viral sweeps, it does not suffice to just record temporally and spatially similar events, that could be just the result of correlation or other \emph{spurious} causes~\cite{aldrich1995correlations,pearl2009causality}.
Specifically, in the context of social influence the spurious causes concern common causes (e.g., assortative mixing  \cite{newman2003mixing} or ``homophily''), unobserved influence, and transitivity of influence~\cite{shalizi2011homophily}.
Other confounding factors (Simpson's paradox~\cite{wagner1982simpson}, temporal clustering~\cite{agarwal2002algorithmic}, screening-off~\cite{hitchcock1997probabilistic}, etc.) make the problem of separating genuine causes from spurious ones even harder.
The problem is intimately connected to the fundamental question in many theories of causality of misinterpreting causation when partitioning so-called genuine vs.~spurious causes, and it
 points to the need for a rigorous foundation, for example, the one developed by such prominent philosophers as Cartwright, Dupr\'e, Skyrms, or Suppes~\cite{hitchcock1997probabilistic}.

In this paper we tackle the problem of deriving, from a database of propagation traces (i.e., traces left by entities flowing in a social network), a set of directed acyclic graphs (\DAGs), each of which representing a genuine causal process underlying the network: multiple causal processes might involve \emph{different communities}, or represent \emph{the backbone of information propagation for different topics}.

Our approach builds upon \emph{Suppes' theory of probabilistic causation}~\cite{suppes_prima_facie}, a theory grounded on a logical foundation that incorporates time, logic and probability, having a long history in philosophy, and whose empirical effectiveness has been largely demonstrated in numerous contexts~\cite{caravagna2016algorithmic,bonchi2017exposing,gao2017efficient}.
Although exhibiting various known limitations, such a theory can be expressed in probabilistic computational tree logic with efficient model checkers that allow for devising efficient learning algorithms~\cite{kleinberg2010algorithmic}. 
All of this makes it particularly appealing given the computational burden of the problem we tackle in this work, and thus preferable to other more sophisticated yet computationally heavier theories, such as 
the one advocated by Judea Pearl and indirectly related to the philosophical foundations laid by Lewis and Reichenbach~\cite{pearl2009causality}.
The central notion behind Suppes' theory is \emph{prima facie causes}: to be recognized as a cause, a certain event must occur before the effect (\emph{temporal priority}) and must lead to an increase of the probability of observing the effect (\emph{probability raising}).
In the context of social influence the latter means that the most influential users in the network heavily influence the behavior of their social ties.

\spara{Challenges and contributions.}  
Suppes' and other similar notions of causality suffer from a well-known weakness of misinterpreting causation when the data spans a mixture of several subtypes, i.e., ``Simpson's paradox''~\cite{aldrich1995correlations}. This issue is predominant in information-propagation analysis, as users have different interests, and different topics produce different propagation traces, affecting very diverse populations. Hence, analyzing the social dynamics as a whole, by modeling the propagation traces altogether, will most likely end up in misinterpreting causation.
In this work we tackle this problem by partitioning the input propagation traces and assigning a different causal interpretation for each set of the identified partition.

More in detail, our goal is to partition the propagation traces into groups, in such a way that each group is as minimally contradictory as possible in terms of the hierarchical structure of information propagation. For this purpose, we borrow the notion of ``\emph{\agony}'', introduced by Gupte~\emph{et~al.}~\cite{GupteSLMI}, as a measure of how clearly defined a hierarchical  structure is in a directed graph.
We introduce the \clusteringprob problem, where the input propagation traces are partitioned into groups exhibiting small agony. We prove that \clusteringprob is \NPhard, and devise efficient algorithms with provable approximation guarantees.

For each group identified in the partitioning step,
the part of social network spanned by the union of all propagation traces in every such group may be interpreted as a \emph{prima-facie} graph, i.e., a graph representing all causal claims (of the causal process underlying that group) that are consistent with the Suppes' notion of prima facie causes.
As broadly discussed in the literature~\cite{hitchcock1997probabilistic}, prima facie causes may be either genuine or spurious.
Therefore, from the prima-facie graph we still need to filter out spurious claims.
We accomplish this second step of our approach by selecting the minimal set of arcs that are the most explanatory of the input propagation traces within the corresponding group.
We cast this problem as a constrained maximum likelihood estimation (MLE) problem. The result is a minimal causal structure (a \DAG) for each group, representing all the genuine causal claims of the social-influence causal process underlying that group.

We present an extensive experimental evaluation on synthetic data with a known ground-truth generative model, and show that our method achieves high accuracy in the task of identifying genuine causal arcs in the ground truth.
On real data we show that our approach can improve subsequent tasks that make use of social  influence, in particular in the task of predicting influence spread.

To summarize, the main contributions of this paper are: 
\squishlist

\item We adopt a principled causal approach to the analysis of social influence from information-propagation data, following Suppes' probabilistic causal theory.

\item We introduce the \clusteringprob problem, where the input set of propagations (\DAGs) is partitioned into groups exhibiting a clear hierarchical  structure. We prove the \NPhard{ness} of the problem and devise efficient algorithms with approximation guarantees. For each resulting group of propagation traces we apply a constrained MLE approach to ultimately learn a minimal causal topology.

\item Experiments on synthetic data show that our method is able to retrieve the genuine causal arcs with respect to a known ground-truth generative model. Experiments on real data show that, by focusing only on the causal structures extracted instead of the whole social network, we can improve the effectiveness of predicting influence spread.
\squishend
The rest of the paper is organized as follows.
Section~\ref{sec:related} overviews the related literature.
Section~\ref{sec:framework} discusses input data and background notions.
Section~\ref{sec:agony_bounded_clustering_problem} defines our problem, by formulating and theoretically characterizing the \clusteringprob problem, and discussing the minimal-causal-topology learning problem.
Section~\ref{sec:algorithms} describes the approximation algorithms for \clusteringprob.
Section~\ref{sec:experiments} presents the experimental evaluation, while Section~\ref{sec:conclusions} concludes the paper.

%% file: rw.tex

A large body of literature has focused on empirically analyzing the effects and the interplay of social influence and other factors of correlation, such as homophily.
Crandall \emph{et al.} \cite{CrandallCHKS08} present a study over Wikipedia editors' social network and LiveJournal blogspace, showing that there exists a feedback effect between users' similarity and social influence, and that combining features based on social ties and similarity is more predictive of future behavior than either social influence or similarity features alone.
Cha \emph{et al.} \cite{ChaMG09} analyze information dynamics on the Flickr social network and  provide empirical evidence that the social links are the dominant method of information propagation, accounting for more than 50\% of the spread of favorite-labeled pictures.
Leskovec \emph{et al.}~\cite{LeskovecSK06,LeskovecAH07} show patterns of influence by studying person-to-person recommendation for purchasing books and videos, while finding conditions under which such recommendations are successful. 
Hill \emph{et al.}~\cite{hill} analyze the adoption of a new telecommunications service and show that it is possible to predict with a significant confidence whether customers will sign up for a new calling plan once one of their phone contacts does so.

Additional effort has been devoted to methods for distinguishing genuine social influence from homophily and other external factors.
Anagnostopoulos \emph{et al.} \cite{aris08} devise techniques (e.g., \emph{shuffle test} and \emph{edge-reversal} test) to separate influence from correlation, showing that in Flickr, while there is substantial social correlation in tagging behavior, such correlation cannot be ascribed to influence.
Aral \emph{et al.} \cite{aral2009distinguishing} develop a matched sample
estimation framework, which accounts for homophily effects
as well as influence. 
Fond and Neville \cite{FondN10}  present a randomization technique
for temporal-network data where the attributes and links
change over time.
Sharma~and~Cosley~\cite{sharma2016distinguishing} propose a statistical procedure to discriminate between social influence and personal preferences in online activity feeds.

Our work distinguishes itself from this literature as it adopts a principled causal approach to the analysis of social influence from information-propagation data, rooted in probabilistic causal theory.
The idea of adopting Suppes' theory to infer causal structures and represent them into graphical models is not new, but it has been used in completely different contexts, such as cancer-progression analysis~\cite{caprese_causation,capri_causation}, discrimination detection in databases~\cite{bonchi2017exposing}, and financial stress-testing scenarios~\cite{gao2017efficient}.

%% file: framework.tex

\spara{Input data.}
The data we take as input in this work consists of: ($i$) a directed graph $G=(V,A)$ representing a network of interconnected objects and hereinafter informally referred to as the ``social graph'', ($ii$) a set $\E$ of \emph{entities}, and ($iii$) a set $\A$ of \emph{observations} involving the objects of the network and the entities in $\E$. Each observation in $\A$ is a triple $\langle v, \phi, t \rangle$, where $v \in V$, $\phi \in \E$, and $t \in \mathbb{N}^+$, denoting that the entity $\phi$ is observed at node $v$ at time $t$.
For instance, $G$ may represent users of a social network interconnected via a follower-followee relation, entities in $\E$ may correspond to pieces of multimedia content (e.g., photos, videos), and an observation $\langle v, \phi, t \rangle \in \A$ may refer to the event that the multimedia item $\phi$ has been enjoyed by user $v$ at time $t$.      
We assume an entity cannot be observed multiple times at the same node; should this happen, we consider only the first one (in order of time) of such observations. 

The set $\A$ of observations can alternatively be viewed as a database $\D$ of \emph{propagation traces} (or simply \emph{propagations}), i.e., traces left by entities ``flowing'' over $G$. Formally, a propagation trace of an entity $\phi$ corresponds to the subset $\{\langle v, \phi', t\rangle \in \A  \mid \phi' = \phi\}$ of all observations in $\A$ involving that entity. 
Coupled with the graph $G$, the database of propagations corresponds to a set $\D = \{D_{\phi} \mid \phi \in \E\}$ of \emph{directed acyclic graphs} (\DAG{s}), where, for each $\phi \in \E$, $D_{\phi} = (V_{\phi}, A_{\phi})$, $V_{\phi} = \{v \in V \mid \langle v,\phi,t \rangle \in \A\}$, $A_{\phi} = \{(u,v) \in A \mid \langle u,\phi,t_u \rangle \in \A, \langle v,\phi,t_v \rangle \in \A, t_u < t_v\}$. Note that each $D_{\phi} \in \D$  contains no cycles, due to the assumption of time irreversibility. 
In the remainder of the paper we will refer to $\D$ \emph{as a database of propagations} or \emph{as a set of} \DAG{s} interchangeably.  
Also, we assume that each propagation is started at time 0 by a dummy node $\Omega \notin V$, representing a source of information external to the network that is implicitly connected to all nodes in $V$. 
An example of our input is provided in Figure~\ref{fig:example1}. 
Given a set $\mathcal{D} \subseteq \D$ of propagations, $G(\mathcal{D})$ denotes the union graph of all \DAGs\ in $\mathcal{D}$, where the union of two graphs $G_1 = (V_1,A_1)$ and $G_2 = (V_2,A_2)$ is  $G_1 \cup G_2 =  (V_1 \cup V_2 \setminus \{\Omega\}, \{(u,v) \in A_1 \cup A_2 \mid u \neq \Omega, v \neq \Omega\})$.\footnote{For the sake of presentation of the technical details, we assume union graphs not containing the dummy node $\Omega$.}
Note that, although $\mathcal{D}$ is a set of \DAGs, $G(\mathcal{D})$ is \emph{not necessarily} a \DAG.


\begin{figure}[t!]
\vspace{-0mm}
\begin{tabular}{cc}
\hspace{-4mm}
\begin{small}
\begin{tabular}{|c c c|}
\multicolumn{3}{c}{\D}\\
\hline
$v$ & $\phi$ & $t$ \\ \hline
$\Omega$ & $\phi_1$ & 0 \\
$v_2$ & $\phi_1$ & 2 \\
$v_3$ & $\phi_1$ & 4 \\
$v_4$ & $\phi_1$ & 5 \\
$v_5$  & $\phi_1$ & 7 \\ \hline
$\Omega$ & $\phi_2$ & 0 \\
$v_2$ & $\phi_2$ & 1 \\
$v_1$ & $\phi_2$ & 3 \\
$v_5$ & $\phi_2$ & 6\\
$v_7$ & $\phi_2$ & 7\\
$v_6$  & $\phi_2$ & 8\\
$v_3$ & $\phi_2$ & 9\\\hline
$\Omega$ & $\phi_3$ & 0 \\
$v_1$ & $\phi_3$ & 1 \\
$v_2$ & $\phi_3$ & 3 \\
$v_6$ & $\phi_3$ & 5\\
$v_7$  & $\phi_3$ & 7\\
$v_4$  & $\phi_3$ & 8\\ \hline
\end{tabular}
\end{small}
&
\hspace{-6mm}
\begin{tabular}{cc}
$\qquad \qquad G$ & $\qquad D_{\phi_1}$\\
\multicolumn{2}{c}{\includegraphics[width=.35\textwidth]{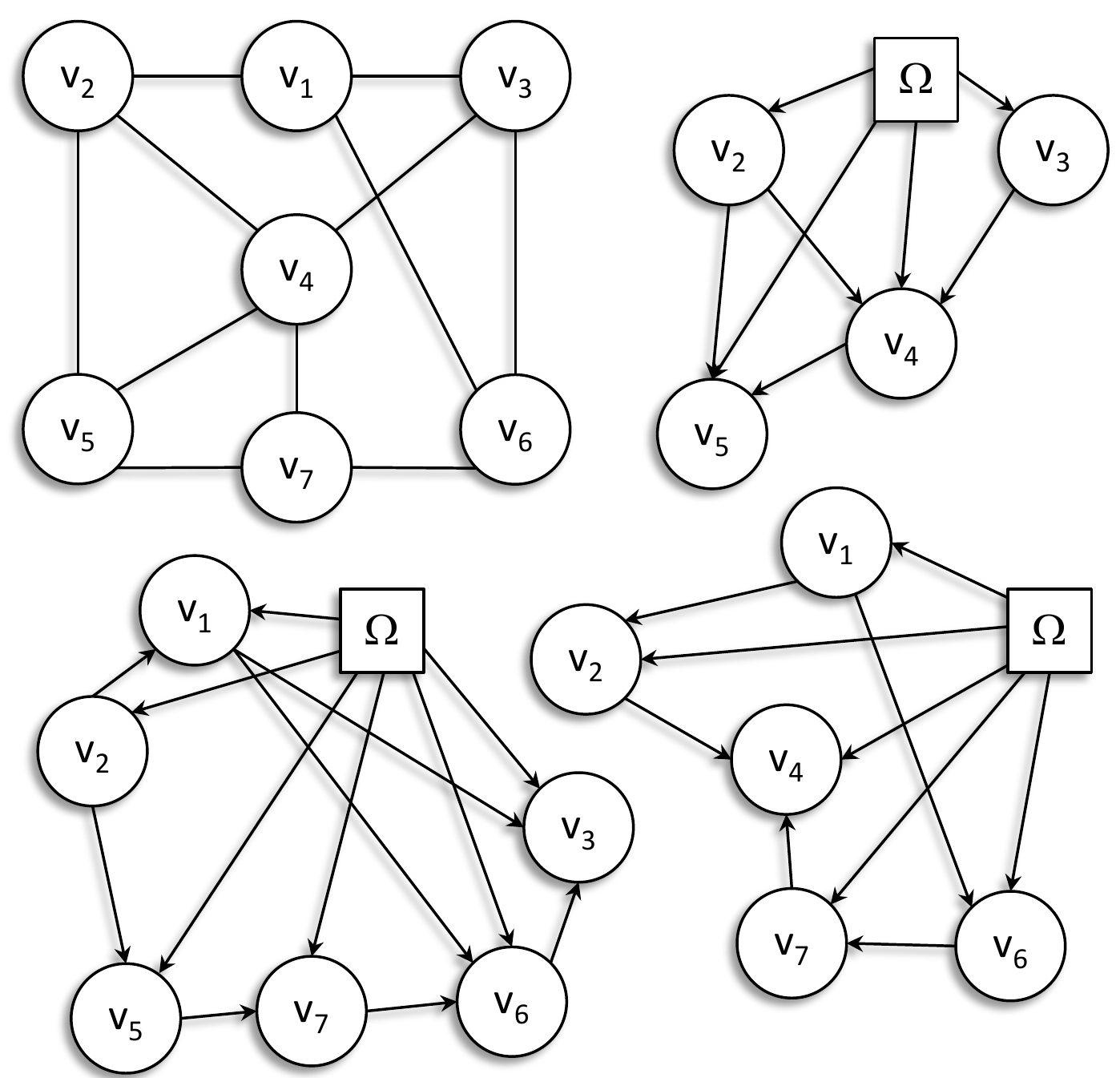}}\\
$\qquad \qquad  D_{\phi_2}$ &$\qquad D_{\phi_3}$\\
\end{tabular}
\end{tabular}
\vspace{-2mm}
\captionsetup{justification=justified,singlelinecheck=off,font={stretch=0.6}}
\caption{\em \small An example of the input of our problem: a social graph $G$, and a database of propagation traces \D\ defined over a set of entities $\E = \{\phi_1, \phi_2, \phi_3\}$. Here $G$ is represented undirected: each edge corresponds to the two directed arcs. Each propagation is started at time 0 by a dummy node $\Omega \notin V$. Given the graph  $G$, the propagation database \D\ is equivalent to the set $\{D_{\phi_1},D_{\phi_2},D_{\phi_3}\}$ of \DAGs. \label{fig:example1}}
\vspace{-2mm}
\end{figure}

\spara{Hierarchical structure.}
As better explained in the next section, in our approach we borrow the notion of ``\emph{\agony}'' introduced by Gupte et al.~\cite{GupteSLMI} to reconstruct a proper hierarchical structure of a directed graph $G = (V,A)$.
Such a notion is defined as follows.
Consider a ranking function $r: V \rightarrow \mathbb{N}$ for the nodes in $V$, such that the inequality $r(u) < r(v)$ expresses the fact that $u$ is ``higher'' in the hierarchy than $v$, i.e., the smaller $r(u)$ is, the more $u$ is an ``early-adopter''. If $r(u) < r(v)$, then the arc $u \rightarrow v$ is expected and does not result in any ``social agony''. If, instead, $r(u) \geq r(v)$ the arc $u \rightarrow v$ leads to agony, because it means that $u$ has a follower $v$ (in the social-graph terminology) that is higher-ranked than $u$ itself. Therefore, given a graph $G$ and a ranking $r$, the \agony\ of each arc $(u,v)$ is defined as $\max\{r(u)-r(v)+1,0\}$, and the \agony\ $a(G,r)$ of the whole graph for a given ranking $r$ is the sum over all arcs:
$
a(G,r) = \sum_{(u,v) \in A}  \max\{r(u)-r(v)+1,0\}.
$
In most contexts (as in ours) the ranking $r$ is not explicitly provided. The objective thus becomes finding a ranking that minimizes the \agony\ of the graph. This way, one can compute the \agony\ of any graph $G$ as
$
a(G) = \min_r a(G,r).
$
As a \DAG\ implicitly induces a partial order over its nodes, it has zero \agony: the nodes of a \DAG\ form a perfect hierarchy. For instance, in the \DAGs\  $D_{\phi_1},D_{\phi_2},D_{\phi_3}$ in Figure~\ref{fig:example1}, it is sufficient to take the temporal ordering as a ranking to get no \agony, i.e., $r(u) = t_u$, where $\langle u,\phi_i,t_u \rangle \in D_{\phi_i}$.
On the other hand, merging several \DAGs\ leads to a graph that is not necessarily a \DAG, and can thus have non-zero \agony. 
In fact, a cycle of  length $k$ (and not overlapping with other cycles) incurs \agony\ equal to $k$~\cite{GupteSLMI}.
A ranking $r$ yielding minimum agony for the union graph $D_{\phi_1} \cup D_{\phi_2}$ of the \DAGs $D_{\phi_1}$ and $D_{\phi_2}$ in Figure~\ref{fig:example1} is $(v_2:0)(v_1:1)(v_4:2)(v_5:3)(v_7:4)(v_6:5)(v_3:6)$.
This ranking yields no \agony\ on all the arcs, except for arc $v_3 \rightarrow v_4$, which incurs \agony\ equal to the length of the cycle involving $v_3$ and $v_4$, i.e., 6-2+1 = 5.

Gupte~et~al.~\cite{GupteSLMI} provide a polynomial-time algorithm for finding a ranking of minimum \agony, running in $\mathcal{O}(|V| \times |A|^2)$ time.
Tatti~\cite{Tatti14} devises a more efficient method, which takes $\mathcal{O}(|A|^2)$ time in the worst case, but recognized as much faster in practice.
In another work~\cite{Tatti17} Tatti also provides methods to compute \agony\ on weighted graphs and when cardinality constraints are specified. 

%% file: agony_clustering.tex

The main goal of this work consists in  deriving \emph{a set of causal \DAGs that are well-representative of the social-influence dynamics underlying an input database of propagation traces}.
The idea is that every derived \DAG represents a  \emph{causal process} underlying the input data, where an arc from $u$ to $v$ models the situation where \emph{actions of user $u$ are likely to cause actions of user $v$}.
The various causal processes may correspond to
the \emph{different topics} on which users exert \emph{social influence} on each other.
For instance, a causal \DAG\ concerning politics may contain an arc $u \rightarrow v$ meaning that user $u$ is influential on user $v$ for politics-related matters, while a \DAG\ for the music context may miss that influence relation or have~it~inverted.

To achieve our goal, we resort to the theory of \emph{probabilistic causation}~\cite{hitchcock1997probabilistic}, which was introduced by Suppes in~\cite{suppes_prima_facie}, via the notion of \emph{prima facie causes}.

\begin{mydefinition}[Prima facie causes~\cite{suppes_prima_facie}] \label{def:praising}
For any two events $c$ (cause) and $e$ (effect), occurring respectively at times $t_c$ and $t_e$,
under the  mild assumption that the probabilities $\Probab{c}$ and $\Probab{e}$ of the two events satisfy the condition $0 < \Probab{c}, \Probab{e} < 1$,
the event $c$ is called a \emph{prima facie cause} of the event $e$ if it occurs \emph{before} $e$ and \emph{raises the probability} of $e$, i.e., $t_c < t_e \ \wedge \ \Pcond{e}{c} > \Pcond{e}{\~ c}$.
\end{mydefinition}

While being a well-acknowledged definition of causality, Suppes' prima facie causes suffer from three main limitations:
($i$) If the input data spans multiple causal processes, causal claims may be hidden or misinterpreted if one looks at the data as a whole: this anomaly is  the well-known Simpson's paradox \cite{wagner1982simpson} .
($ii$) Suppes' definition
lacks any characterization in terms of \emph{spatial proximity}, which is a critical concept in social influence, as users who never interact with each other should intuitively not be involved in a causal relation.
($iii$) As discussed in the Introduction, prima facie causes may be either \emph{genuine} or \emph{spurious}.
The desideratum is that only the former ones are detected and  presented~as~output.


Motivated by the above discussion, we propose to derive the desired social-influence causal \DAGs with a two-step methodology, where the successive execution of the two steps ultimately output social-influence causal \DAGs that are consistent with the principles of Suppes' theory, while at the same time overcoming its limitations.
Specifically, we accomplish the two steps by formulating and solving two problems:
\begin{itemize}
\item
\clusteringprob (Section~\ref{sec:subproblem1}), a novel combinatorial-optimization problem mainly designed to get rid of the Simpson's paradox, where the input propagation set is partitioned into homogeneous groups. Formulating, theoretically characterizing, and solving this problem constitute the main technical contribution of this work.
\item
\minimaltopology (Section~\ref{sec:subproblem2}), a learning problem where a minimal causal \DAG is derived from the union graph of each group of propagations identified in the first step. The main goal here is to remove all the spurious relationships.
\end{itemize}

\subsection{Partitioning the propagation set}
\label{sec:subproblem1}
The input propagation set is typically so large that it may easily span multiple causal processes, each of which corresponds to a different social-influence topic.
Due to the aforementioned Simpson's paradox, attempting to infer causality from all input propagations at once is therefore bound to fail.
For this purpose,  we aim at preventively partitioning the input set of propagations into \emph{homogeneous} groups, each of which is likely to identify a single causal process.
Homogeneity means that the propagations in a group should exhibit as few ``contradictions'' as possible, where a contradiction arises when a user $u$ is a (direct or indirect) influencer for user $v$ in some propagations, while the other way around holds in some other propagations.
In other words, a homogeneous  group of propagations should be such that the union graph of all those propagations has a clear \emph{hierarchical structure}, i.e., it is as similar as possible to a \DAG.
\emph{This requirement is fully reflected in the notion of agony} discussed in Section~\ref{sec:framework}.
For this reason, here we resort to that notion and define the \clusteringprob problem:
given a threshold  $\eta \in \mathbb{N}$, partition the input propagations into the minimum number of sets such that the union graph of every of such sets has agony no more than $\eta$.
Additionally, we require for each set to be limited in size and to exhibit a union graph that is (weakly) connected, as too large propagation sets or disconnected union graphs are unlikely to represent single causal processes.

\begin{problem}[\clusteringprob]\label{prob:clustering}
Given a set $\D$ of \DAGs  and two positive integers $K, \eta \in \mathbb{N}$, find a partition $\mathbf{D}^* \in \mathcal{P}(\D)$ (where $\mathcal{P}(\cdot)$ denotes the set of all partitions of a given set) such that
\begin{eqnarray}
\mathbf{D}^* & = & \operatorname{argmin}_{\mathbf{D} \in \mathcal{P}(\D)} |\mathbf{D}|  \qquad \mbox{subject to} \nonumber\\\
& & \!\!\!\!\!\!\!\!\!\!\!\!\!\!\!\!\!\!\!\!\!\!\!\!\!\!\!\!\!\!\!\!\! \forall \mathcal{D} \in \mathbf{D}: \ a(G(\mathcal{D})) \leq \eta, \ \ |\mathcal{D}| \leq K, \ \ G(\mathcal{D}) \mbox{ is weakly-connected}.\label{eq:clustering-prob-constraints}
\end{eqnarray}
\end{problem}

For every $\mathcal{D} \in \mathbf{D}^*$, we informally term the union graph $G(\mathcal{D})$  \emph{prima-facie graph}.
In fact, apart from  addressing the Simpson's paradox, every union graph $G(\mathcal{D})$ may be interpreted as a graph containing all prima-facie causes underlying the propagation group $\mathcal{D}$.
To this purpose, note that every $G(\mathcal{D})$ reflects both the temporal-priority and the probability-raising conditions in Suppes' theory.
Temporal priority is guaranteed by the input itself, as an arc $u \rightarrow v$ in some input \DAG (and, thus, in every output union graph) exists only if an input observation exists in $\A$, where the same entity is observed first in $u$ and then in $v$.
Probability raising arises as we ask for a partition with the minimum number of small-agony groups, which means that the identified groups will be as large as possible (as long as the problem constraints are satisfied).
This way, every output group will likely contain as much evidence as possible for every causal claim $u \rightarrow v$, i.e., a large evidence that the effect is observed thanks to the cause.
At the same time, the output union graphs overcome the limitation of missing spatial proximity in Suppes' theory: all arcs within the output graphs come from the input social graph $G$, which implicitly encodes a notion of spatial proximity corresponding to the social relationships among users.

Switching to a technical level, a simple observation on  Problem~\ref{prob:clustering} is that it is well-defined, as it always admits at least the solution where every \DAG forms a singleton group.
A more interesting characterization is the  \NPhard{ness} of the problem, which we formally state next.
In Section~\ref{sec:algorithms} we will instead present the approximation algorithms we designed to solve the problem.

\begin{mytheorem}
Problem~\ref{prob:clustering} is \NPhard.
\end{mytheorem}
\begin{proof}
We consider the well-known \NPcomplete\ (decision version of the) \setcover problem~\cite{Cook}:
given a universe $U = \{e_1, \ldots, e_n\}$ of elements, a set $\mathcal{S} = \{S_1, \ldots, S_m\} \subseteq 2^U$ of subsets of $U$, and a positive integer $p$, is there a subset $\mathcal{S}^* \subseteq \mathcal{S}$ of size $\leq p$ covering the whole universe $U$?
We reduce \setcover to the decision version of our \clusteringprob, which is as follows: given a set $\D$ of \DAGs, and three positive integers $K, \eta, q$, is there a partition of $\D$ of size $\leq q$ satisfying all constraints in Equation~(\ref{eq:clustering-prob-constraints})?
Given an instance $I = \langle U, \mathcal{S}, p \rangle$ of \setcover, we construct  an instance $I' = \langle \D, K, \eta, q \rangle$ of \clusteringprob
so that $I$ is a yes-instance for \setcover if and only if $I'$ is a yes-instance for \clusteringprob.
To this end, we set $\eta = m(n+2)-1$, $K = |\D|$, $q = p$, and we let $\D$ be composed of a \DAG $D_i$ for every element $e_i \in U$.
Each \DAG $D_i \in \D$ has node set $V(D_i) = \{u_i, v_1, \ldots, v_m\} \cup  V_1^i \cup  \cdots \cup V_m^i$, i.e.,
it comprises a node $u_i$, nodes $v_1, \ldots, v_m$ such that $v_j$ corresponds to set $S_j \in \mathcal{S}$,
and further node sets $V_1^i, \ldots, V_m^i$ defined as
\vspace{-2mm}
$$
V_j^i = \left \{
\begin{array}{ll}
V_j \setminus \{v_{ji}\}, & \ \ \mbox{if} \ e_i \notin S_j,\\
V_j \setminus \{v_{j0}, v_{ji}\}, & \ \ \mbox{otherwise},
\end{array}\right .
\vspace{-2mm}
$$
where $V_j = \{v_{j0}, v_{j1}, \ldots, v_{jn}\}$.
The arc set of $D_i$ is composed of an arc from $u_i$ to every node $v_j$, and, for each $v_j$, ($i$) an arc $(v_j, v_{j0})$ (only if $v_{j0} \in V_j^i$), ($ii$) an arc $(v_{jn}, v_j)$ (only if $v_{jn}\in V_j^i$), and ($iii$) an arc $(v_{jk}, v_{jk+1})$ for every $k \in [0..n-1]$ such that $v_{jk}, v_{jk+1} \in V_j^i$.
It is easy to see that each $D_i$ can be constructed in polynomial time in the size of $I$, and is actually a \DAG.
In this regard, note that, for each $v_j$, $V_j^i$ misses at least one node needed to form the cycle $v_j \rightarrow v_{j0} \rightarrow \cdots \rightarrow v_{jn} \rightarrow v_j$ (i.e., at least the node $v_{ji}$).

Let $\mathbf{D}^+ \subseteq 2^{\D}$ be the set of all \DAG sets that are admissible for \clusteringprob on instance $I'$.
We now show that $\mathbf{D}^+ \equiv \mathcal{S}^+$, where $\mathcal{S}^+ = \bigcup_{S \in \mathcal{S}} 2^{S}$ corresponds to the set $\mathcal{S}$ augmented by all possible subsets of every set within it.
First of all, we observe that every \DAG set $\mathcal{D} \subseteq \D$ meets both the constraint on the connectedness of the union graph $G(\mathcal{D})$ (each \DAG is connected and shares at least one node with every other \DAG in $\D$), and the constraint on the maximum size (as $K = |\D|$).
This way, the set $\mathbf{D}^+$ of admissible \DAGs is determined by the agony constraint only.
For every $j \in [1..m]$, the union graph $G(\mathcal{D})$ of a \DAG set $\mathcal{D} \subseteq \D$ contains a cycle  $v_j \rightarrow v_{j0} \rightarrow \cdots \rightarrow v_{jn} \rightarrow v_j$ only if set $S_j$ does not contain all elements (corresponding to the \DAGs) in $\mathcal{D}$.
Denoting by $m_{\mathcal{D}}$ the number of such cycles being present in $G(\mathcal{D})$, the agony of $G(\mathcal{D})$ is equal to $m_{\mathcal{D}}(n+2)$, as
all cycles have length $n\!+\!2$ and they are all node-disjoint.
Therefore, the agony of $G(\mathcal{D})$ exceeds the threshold $\eta$ only if $m_{\mathcal{D}} = m$, i.e., only if there exists no set in $\mathcal{S}^+$ containing $\mathcal{D}$, thus meaning that $\mathbf{D}^+ \equiv \mathcal{S}^+$.

We are now ready to show the desired ``$\Leftrightarrow$'' implication between instances $I$ and $I'$.
As for the ``$\Rightarrow$'' part, we notice that, from the covering $\mathcal{S}^*$ representing the solution of \setcover on $I$, it can always be derived a covering $\mathcal{S}'$ of size $|\mathcal{S}'| = |\mathcal{S}^*|$ that is a partition of $U$ (e.g., by processing each $S \in \mathcal{S}^*$ following some ordering, and replacing any set $S$ containing a subset $S' \subseteq S$ of elements covered by already-processed sets with the set $S \setminus S'$).
$\mathcal{S'}$ is a solution of \clusteringprob on $I'$ as its size is $|\mathcal{S}'| = |\mathcal{S}^*| \leq p = q$, it is a partition of $U$ (and, therefore, of $\D$), and it is admissible as $\mathcal{S}' \subseteq \mathcal{S}^+$.
The ``$\Leftarrow$'' part holds due to the following reasoning.
Given the solution $\mathbf{D}^*$ of \clusteringprob on $I'$, one can derive a covering $\mathbf{D}'$ where every set $\mathcal{D} \in \mathbf{D}^*$ such that $\mathcal{D} \in \mathcal{S}^+ \setminus \mathcal{S}$ is replaced with the original set $S_{\mathcal{D}} \in \mathcal{S}$ where $\mathcal{D}$ has been derived from.
$\mathbf{D}'$ is a solution of \setcover on $I$ as its size is $|\mathbf{D}'| = |\mathbf{D}^*| \leq q = p$, it covers $U$, and it is a subset of $\mathcal{S}$.
\end{proof}


\subsection{Learning the minimal causal topology}
\label{sec:subproblem2}

As mentioned above, prima facie causes may be genuine or spurious~\cite{hitchcock1997probabilistic}.
Thus, from the prima-facie graphs identified in the previous step we still need to remove the spurious relationships.
To this end, we aim at identifying a \emph{minimal causal topology} for each prima-facie graph, i.e., selecting  the minimal set of arcs that best explain the input propagations spanned by that graph.

Given a partition $\mathbf{D}^*$ of the input database $\D$ of propagations (computed in the previous step), we first reconstruct a \DAG $G_D(\mathcal{D})$ from the prima-facie graph $G(\mathcal{D})$ of every group $\mathcal{D} \in \mathbf{D}^*$.
To this end, we exploit the by-product of agony minimization on $G(\mathcal{D})$, i.e., a ranking $r$ of the nodes in $G(\mathcal{D})$ (see Section~\ref{sec:framework}).
Specifically, we build $G_D(\mathcal{D})$ by taking all and only those arcs of $G(\mathcal{D})$ that are in accordance with $r$, i.e., all arcs $(u,v)$ such that $r(u) < r(v)$.
%
Then, for every reconstructed \DAG $G_D(\mathcal{D})$, we learn its minimal causal topology via (constrained) \emph{maximum likelihood estimation} (MLE), where the arcs of $G_D(\mathcal{D})$ maximizing a likelihood score such as \emph{Bayesian  Information Criterion}~(BIC)~\cite{bic_1978} or \emph{Akaike Information Criterion}~(AIC)~\cite{akaike1974new} are identified (we experimented with both criteria, see Section~\ref{sec:experiments}).
More precisely, given a database $\D$ of propagations and a set $\hat{A} \subseteq A$ of arcs, we define
$$
f(\hat{A}, \D) = LL(\D|\hat{A}) - \mathcal{R}(\hat{A}),
$$
where $LL(\cdot)$ is the log-likelihood, while $\mathcal{R}(\cdot)$ is a regularization~term.
The \DAG induced by $\hat{A}$ in turn induces a probability distribution over its nodes $\{u_1, \ldots, u_n\}$:
$$
\Probab{u_1, \ldots, u_n} = \prod_{u_i=1}^n \Pcond{u_i}{\pi_i}, \qquad
\Pcond{u_i}{\pi_i} = \bth_{u_i\mid \pi_i},
$$
where $\pi_i = \{ u_j \mid u_j \to u_i \in \hat{A}\}$ are $u_i$'s parents in the \DAG, and $\bth_{u_i\mid \pi(u_i)}$ is a probability density function. Then, the log-likelihood of the network is defined as:
$$
LL(\D|\hat{A}) = \log\Pcond{\D}{\hat{A},\bth}.
$$
The regularization~term $\mathcal{R}(\hat{A})$ introduces a penalty term for the number of parameters in the model and the size of the data. Specifically, $S$ being the number of samples, $\mathcal{R}(\hat{A})$ is defined as $|\hat{A}|$ for AIC and $\frac{|\hat{A}|}{2}\log S$ for BIC.

The problem we tackle here is formally stated as follows.
\begin{problem}[\minimaltopology]\label{problem2}
Given a database $\D$ of propagations and a \DAG $G_D(\mathcal{D}) = (V_D,A_D)$, find $A_D^*(\mathcal{D}) = \arg\max_{\hat{A}_D \subseteq A_D} f(\hat{A}_D, \D)$.
\end{problem}



Even if constrained (the output arc set $A_D^*$ must be a subset of $A_D$), Problem~\ref{problem2} can easily be shown to be still \NPhard~\cite{learning_NP_hard}.
Therefore, we solve the problem by a classic greedy hill-climbing heuristic, whose effectiveness has been well-recognized~\cite{koller2009probabilistic}.

The ultimate output of our second step (and of our overall approach) is a set of  \DAGs $\{G^*_D(\mathcal{D})\}_{\mathcal{D} \in \mathbf{D}^*}$, where each $G^*_D(\mathcal{D})$ is the \DAG identified by the arc set $A^*(\mathcal{D})$ as defined in Problem~\ref{problem2}.
Every $G^*_D(\mathcal{D})$ is a causal \DAG representative of a specific social-influence causal process underlying the input propagation database $\D$.

%% file: algorithms.tex

%

In this section we focus on the algorithms for the \clusteringprob problem (Problem~\ref{prob:clustering}).
Due to its \NPhard{ness}, we clearly cannot aim at optimality, and focus instead on the design of effective and efficient approximation algorithms.
Specifically, we first show how a simple two-step strategy solves the problem with provable guarantees.
This method however suffers from the limitation that the first step is exponential in the size of the input \DAG set, which considerably limits  its applicability in practice.
We hence design a more refined sampling-based algorithm, which still comes with provable  guarantees, while also overcoming the exponential blowup.

\spara{A simple two-step algorithm.}
An immediate solution to Problem~\ref{prob:clustering} consists in first computing all subsets of the input \DAG set $\D$ that satisfy the constraints on agony, size, and connectedness listed in Equation~(\ref{eq:clustering-prob-constraints}) (Step~I), and then taking a minimum-sized  subset of these valid \DAG sets  that is a partition of $\D$ (Step~II).

Step~I can be solved by resorting to frequent-itemset mining \cite{CharuBook}. Specifically, in our setting the \DAGs in $\D$ correspond to items and the support of a \DAG set (itemset) $\mathcal{D}$ is given by the agony of the union graph $G(\mathcal{D})$.
It is easy to see that the constraint on agony is monotonically non-decreasing as the size of a \DAG set increases, i.e., $a(G(\mathcal{D}')) \leq a(G(\mathcal{D}''))$ for any two \DAG sets $\mathcal{D}' \subseteq \mathcal{D}''$.
This way, any downward-closure-based algorithm for frequent itemset-mining (e.g., Apriori \cite{apriori}) can easily be adapted to mine all \DAG sets satisfying the agony constraint. 
The two additional constraints on (1) connectedness and (2) size can easily be fulfilled by (1) filtering out all mined \DAG sets that are not connected,
and (2) stopping the mining procedure once the maximum size $K$ has been reached. 

\begin{algorithm}[t]
\caption{\twostepalg}
\label{alg:two-step-clustering}
\begin{algorithmic}[1]
\small
\REQUIRE{A set $\D$ of \DAGs; two positive integers $K$, $\eta$}
\ENSURE{A partition $\mathbf{D}^*$ of $\D$}
\vspace{1mm}
\STATE{$\mathbf{D}^+ \gets \mbox{\textsf{Mine-Valid-}\DAG-\textsf{sets}}(\D,K,\eta)$}
\STATE{$\mathbf{D}^* \gets \mbox{\textsf{Greedy-Set-Cover}}(\mathbf{D}^+)$}
\end{algorithmic}
\end{algorithm}

As for Step~II, we observe that solving \setcover~\cite{Cook} on the set $\mathbf{D}^+ \subseteq 2^{\D}$ mined in Step~II gives the optimal solution to Problem~\ref{prob:clustering} too.
Therefore, we solve Step~II by the well-known \setcover greedy algorithm which iteratively brings to the solution the set having the maximum number of still uncovered elements~\cite{Cook}, and has approximation factor logarithmic in the maximum size of an input set.
Note that, as $\mathbf{D}^+$ contains all subsets of every set $\mathcal{D} \in \mathbf{D}^+$, at each step of the greedy \setcover algorithm there are multiple sets maximizing the number of uncovered elements, all of them being equivalent in terms of soundness and approximation ratio.
Among them, we therefore choose the one having no already-covered elements. This way, the output covering is guaranteed to be a partition of $\D$, as required by \clusteringprob.

The outline of this simple two-step method is reported as Algorithm \ref{alg:two-step-clustering}.
The algorithm achieves a $\log K$ approximation ratio.

\begin{mytheorem}\label{th:approx-twostep}
$\!\!$\mbox{Algorithm~\ref{alg:two-step-clustering}~is~a~$(\log K)$-approximation~for~Problem~\ref{prob:clustering}}.
\end{mytheorem}
\begin{proof}
Step~I of Algorithm~\ref{alg:two-step-clustering} computes all  \DAG sets $\mathbf{D}^+$ meeting the constraints of Problem~\ref{prob:clustering}.
For every set $\mathcal{D} \in \mathbf{D}^+$, $\mathbf{D}^+$ contains all subsets $\mathcal{D}' \subseteq \mathcal{D}$ too.
This ensures that, for every feasible solution $\hat{\mathcal{S}}$ of \setcover on  $\mathbf{D}^+$, there exists a \setcover solution $\hat{\mathcal{S}}'$ being a partition of $\D$ and having the same objective-function value as $\hat{\mathcal{S}}$.
Thus, for any $\beta$-approximation solution of \setcover on $\mathbf{D}^+$, there exists a $\beta$-approximation solution that is a partition of $\D$.  
The traditional greedy approximation algorithm for \setcover has an approximation factor proportional to the logarithm of the maximum size of an input set.
Hence, running such a greedy method on input $\mathbf{D}^+$ gives approximation guarantees of $\log K$, as all sets in $\mathbf{D}^+$ have size no more than $K$, due to Problem~\ref{prob:clustering}'s constraints. 
\end{proof}

\begin{algorithm}[t]
\caption{\proposedalg}
\label{alg:proposed-sampling}
\begin{algorithmic}[1]
\small
\REQUIRE{A set $\D$ of \DAGs; two positive integers $K$, $\eta$; a real number $\alpha \in (0,1]$}
\ENSURE{A partition $\mathbf{D}^*$ of $\D$}
\vspace{1mm}
\STATE{$\mathbf{D}^* \gets \emptyset$}, \ \ $\D_u \gets \D$
\WHILE{$|\D_u| > 0$}
\vspace{1mm}
\STATE $\mathcal{D}_{s} \gets \emptyset$
\WHILE{$|\mathcal{D}_{s}| < \lceil \alpha\times\min\{K,|\D_u|\} \rceil$}
	\STATE{$\mathcal{D}_{s} \gets \mbox{\samplingSubroutine}(\D_u, K, \eta)$}\hfill\COMMENT{Algorithm~\ref{alg:samplingSubroutine}}
\ENDWHILE
\STATE $\mathbf{D}^* \gets \mathbf{D}^* \cup \{\mathcal{D}_{s}\}$, \ \ $\D_u \gets \D_u \setminus \mathcal{D}_{s}$
\ENDWHILE
\end{algorithmic}
\end{algorithm}

\begin{algorithm}[t]
\caption{\samplingSubroutine}
\label{alg:samplingSubroutine}
\begin{algorithmic}[1]
\small
\REQUIRE{A set $\D_u$ of \DAGs; two positive integers $K$, $\eta$}
\ENSURE{$\mathcal{D}_{s} \subseteq \D_u$}
\vspace{1mm}
\STATE $\mathcal{D}_{s} \gets \emptyset$, \ \ $G(\mathcal{D}_{s}) \gets $ empty graph, \ \ $\D'_u \gets \D_u$
\WHILE{$|\mathcal{D}_{s}| < \min\{K,|\D_u|\} \wedge a(G(\mathcal{D}_{s})) \leq \eta$}
	\STATE $\D_s \gets \emptyset$
	\FORALL{$D \in \D'_u$}
		\STATE $G(\mathcal{D}'_{s}) \gets G(\mathcal{D}_s) \cup D$
		\IF{$G(\mathcal{D}'_{s})$ is weakly connected}		
			\STATE $a(G(\mathcal{D}'_{s})) \gets$ \textsf{Compute-Agony}($G(\mathcal{D}'_{s}$))\hfill\COMMENT{cf.~\cite{Tatti14}}
			\STATE {\bf if} {$a(G(\mathcal{D}'_{s})) \leq \eta$} \ {\bf then} \ $\D_s \gets \D_s \cup \{D\}$
			\STATE {\bf else} \ $\D'_u \gets \D'_u \setminus \{D\}$
		\ENDIF
	\ENDFOR
	\STATE $D^* \gets $ sample a \DAG from $\D_s$
	\STATE $\mathcal{D}_{s} \gets \mathcal{D}_{s} \cup \{D^*\}$, \ \ $G(\mathcal{D}_{s}) \gets G(\mathcal{D}_{s}) \cup D^*$, \ \ $\D'_u \gets \D'_u \setminus \{D^*\}$
\ENDWHILE
\end{algorithmic}
\end{algorithm}

\spara{A sampling-based algorithm.}
The algorithm described above is easy-to-implement and comes with provable approximation guarantees.
Nevertheless, it has a major drawback that the first step is intrinsically exponential in the size of the input \DAG set.
Even though pruning techniques can be borrowed from the frequent-itemset-mining domain, there is no guarantee in practice that the algorithm always terminates in reasonable time.
For instance, when the size of the \DAG sets satisfying the input constraints  tends to be large, the portion of the lattice to be visited may explode regardless of the pruning power of the specific method.
This is a well-recognized issue of frequent-itemset mining \cite{CharuBook}.

Faced with this hurdle, we devise an advanced algorithm that deals with the pattern-explosion issue, while still achieving approximation guarantees.
The proposed algorithm, termed \proposedalg and outlined as Algorithm~\ref{alg:proposed-sampling}, follows a greedy scheme and has a parameter $\alpha \in (0,1]$ to trade off between accuracy and efficiency.
The algorithm iteratively looks for a \emph{maximal} admissible \DAG set $\mathcal{D}_{s}$ covering a number of still uncovered \DAGs no less than $\alpha \times \min \{K, |\D_u| \}$, where $\D_u$ is the set of all still uncovered \DAGs (Lines~4--5).
$\mathcal{D}_{s}$ is computed by repeatedly sampling the lattice of all admissible (and not yet covered) \DAG sets, until a \DAG set satisfying the requirement has been found.
Sampling can be performed by, e.g., uniform \cite{AiZaki2009} or random \cite{MoensGoethals2013} maximal frequent-itemset sampling.
In this work we use the latter.
The outline of the sampling subroutine is in Algorithm~\ref{alg:samplingSubroutine}.
That procedure takes the set $\D_u$ of uncovered \DAGs and selects \DAGs until $\D_u$ has become empty, or the max size $K$ has been reached, or the agony constraint on the union graph of the current \DAG set has been violated (Lines~2--11).
To select a \DAG, the subset $\D_s \subseteq \D'_u$ of admissible \DAGs is first built by retaining all \DAGs that meet constraints on  connectedness and agony if added to the current union graph (Lines~3--9). Note that, if a \DAG violates the agony constraint, it cannot become admissible anymore, thus it is permanently discarded (Line~9).
The same does not hold for the connectedness constraint.

\proposedalg can be proved to be a $\frac{\log K}{\alpha}$-approximation algorithm for \clusteringprob, as formally stated in Theorem~\ref{th:approx-proposed}.
Thus, parameter $\alpha$ represents a knob to trade off accuracy vs. efficiency: a larger $\alpha$ gives a better approximation factor (thus, better accuracy), but, at the same time, leads to bigger running time as more sampling iterations are needed to find a \DAG set meeting a more strict constraint.

\begin{mytheorem}\label{th:approx-proposed}
$\!$\mbox{Algorithm~\ref{alg:proposed-sampling} is a $\frac{\log K}{\alpha}$-approximation for Problem~\ref{prob:clustering}}.
\end{mytheorem}
\begin{proof}
For each step $t$, let $u_{max}(t)$ denote the maximum number of uncovered elements that can be covered by a set not yet included in the current solution.
It is known that a greedy algorithm for \setcover  finding at each step a set that covers a fraction of uncovered elements no less than $\alpha \times u_{max}(t)$ achieves $\frac{\log h_{max}}{\alpha}$ approximation guarantees, where $h_{max}$ is the maximum size of an input set (see Lemma 2 in \cite{Cormode2010}).
In our context $h_{max} \leq K$, $u_{max}(t) = \min\{K,|\D_u(t)|\}$, and the \DAG set  $\mathcal{D}_{s}(t)$ that is added to the solution by Algorithm~\ref{alg:proposed-sampling} at Step~$t$ covers a number of still uncovered \DAGs that is $\geq \alpha \times \min \{K, |\D_u(t)|\} = \alpha \times u_{max}(t)$. Thus, the ultimate approximation factor of Algorithm \ref{alg:proposed-sampling} is $\frac{\log K}{\alpha}$.
\end{proof}


\spara{Time complexity.}
Let $H$ be the (maximum) number of sampling iterations needed to find a valid \DAG set $\mathcal{D}_s$.
The sampling subroutine (Algorithm~\ref{alg:samplingSubroutine}) takes $\mathcal{O}(K~|\D|~T_a)$ time, where $T_a$ is the (maximum) time spent for a single agony computation.
The subroutine is executed $\mathcal{O}(H~|\mathbf{D}^*|)$ times in Algorithm~\ref{alg:proposed-sampling}.
Thus, the overall time complexity of the proposed \proposedalg algorithm is $\mathcal{O}(H~|\mathbf{D}^*|~K~|\D|~T_a)$.
The efficient agony-computation method in~\cite{Tatti14} takes time quadratic in the edges of the input graph.
Hence, $T_a$ is bounded by $\mathcal{O}(|A|^2)$, where $A$ is the arc set of the input social graph $G$.
However, this is a very pessimistic bound, first because, as remarked by the authors themselves, the method in~\cite{Tatti14} is much faster in practice, and, more importantly, because agony computation in Algorithm~\ref{alg:samplingSubroutine} is run on much smaller subgraphs of~$G$.

\spara{Implementation.}
A number of expedients may be employed to speed up the \proposedalg algorithm in practice, including:
($i$)
Preventively discard input \DAGs violating the connectedness constraint with all other \DAGs.
($ii$)
Algorithm~\ref{alg:samplingSubroutine}, Lines~5~and~7: check whether $D$ has arcs spanning nodes already in $G(\mathcal{D}_s)$; if not, skip agony computation of $G(\mathcal{D}'_s)$ (and set it equal to $a(G(\mathcal{D}_s))$). 
($iii$)
Approximate agony instead of computing it exactly (by, e.g., allowing only a fixed amount of time to the anytime algorithm in~\cite{Tatti14}); correctness of the algorithm is not affected as the approximated agony is an upper bound on the exact value.
($iv$)
Adopt a \emph{beam-search} strategy: sample $\widetilde{\D}_u$ from $\D_u$ (with $|\widetilde{\D}_u| < |\D_u|$, e.g., $|\widetilde{\D}_u| = \mathcal{O}(\log |\D_u|)$) and use $\widetilde{\D}_u$ in Algorithm~\ref{alg:samplingSubroutine} instead of $\D_u$.
($v$)
Repeat the sampling procedure at Lines~4--5 of Algorithm~\ref{alg:proposed-sampling}  for a maximum number of iterations; after that, sample a \DAG $D$ from $\D_u$, and add  $\{D\}$ to the solution. 
($vi$)
Run Algorithm~\ref{alg:proposed-sampling} until $|\D_u| > \epsilon$; after that, add the \DAGs still in $\D_u$ as singletons to the solution.

%% file: experiments_new.tex

In this section we show the performance of the proposed method on both synthetic data, where the true generative model is given, as well as on real data, where no ground-truth is available.

\spara{Reproducibility.}
Code and datasets available at \href{http://bit.ly/2BEV5k9}{bit.ly/2BEV5k9}.

\subsection{Synthetic data} \label{sec:synthetic_data}
\spara{Generation.}
We employ the following methodology:
\begin{enumerate}
\item Randomly generate a directed graph $G = (V,A)$ (our social graph) with $n$ nodes and density $\delta$ (where by density we mean number of edges divided by $\binom{n}{2}$).
\item Randomly partition the node set $V$ of the generated graph $G$ into a set of $k$ groups $\{S_1,\ldots,S_k\}$ such that: $(i)$ $\forall i \in [1,k]$, $G(S_i)$ is weakly-connected, $(ii)$ $\forall i \in [1,k]: |S_i| \in [\mbox{card}_{min},\mbox{card}_{max}] \in (0,n)$, and $(iii)$ the number of overlapping nodes between any pair of groups is bounded by $\mbox{card}_{overlap} \leq \mbox{card}_{max}$.
\item Generate a causal \DAG\ $G_{cause_i}$ of density $\delta_{G_{cause}}$ for each of the $k$ groups. Namely, we obtain $\mathbb{G}_{cause} = \{G_{cause_1},\ldots,G_{cause_k}\}$, by considering the union graph of each partition independently and then randomly removing arcs from it in order to discard cycles. 
Then, we assign each remaining edge a random value bounded by $[p_{min},p_{max}]$ (where $p_{min},p_{max} \in (0,1)$), representing the conditional probability of the child node given the connected parent (we assume the probability of any child given the absence of all its parents to be $0$).
\item Generate a set $\A_k$ of observations for each of the $k$ groups in terms of triples $\langle v, \phi, t \rangle$ (with $\bigcup_{k} \A_k$ corresponding to the whole input set $\A$ of observations), such that each user performs at least one action, but, at the same time, she does not perform an action on at least one item of each cluster. To do this, we sample a set of traces with probability distributions induced by the causal \DAGs; in this way, we obtain a set of users/entities that are observed in the trace. We then add the time $t$ to each observation randomly, but still consistently with the total ordering defined by the \DAG.
\end{enumerate}

\begin{figure*}[t]
\centering
\vspace{-4mm}
\begin{tabular}{@{\!\!\!\!}@{ }@{ }c@{ }@{ }c@{ }@{ }c@{}}
\includegraphics[width=.33\textwidth]{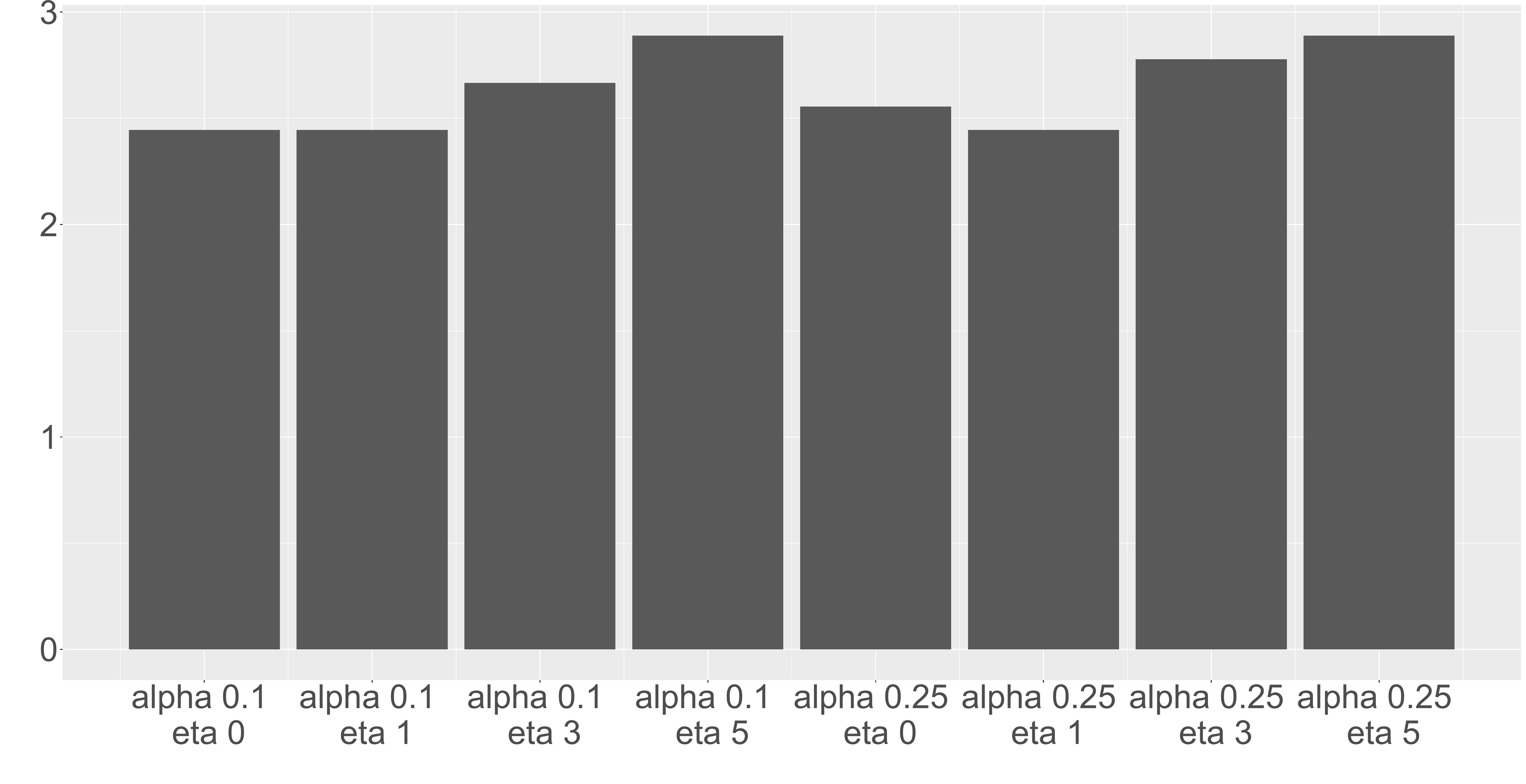} & \includegraphics[width=.33\textwidth]{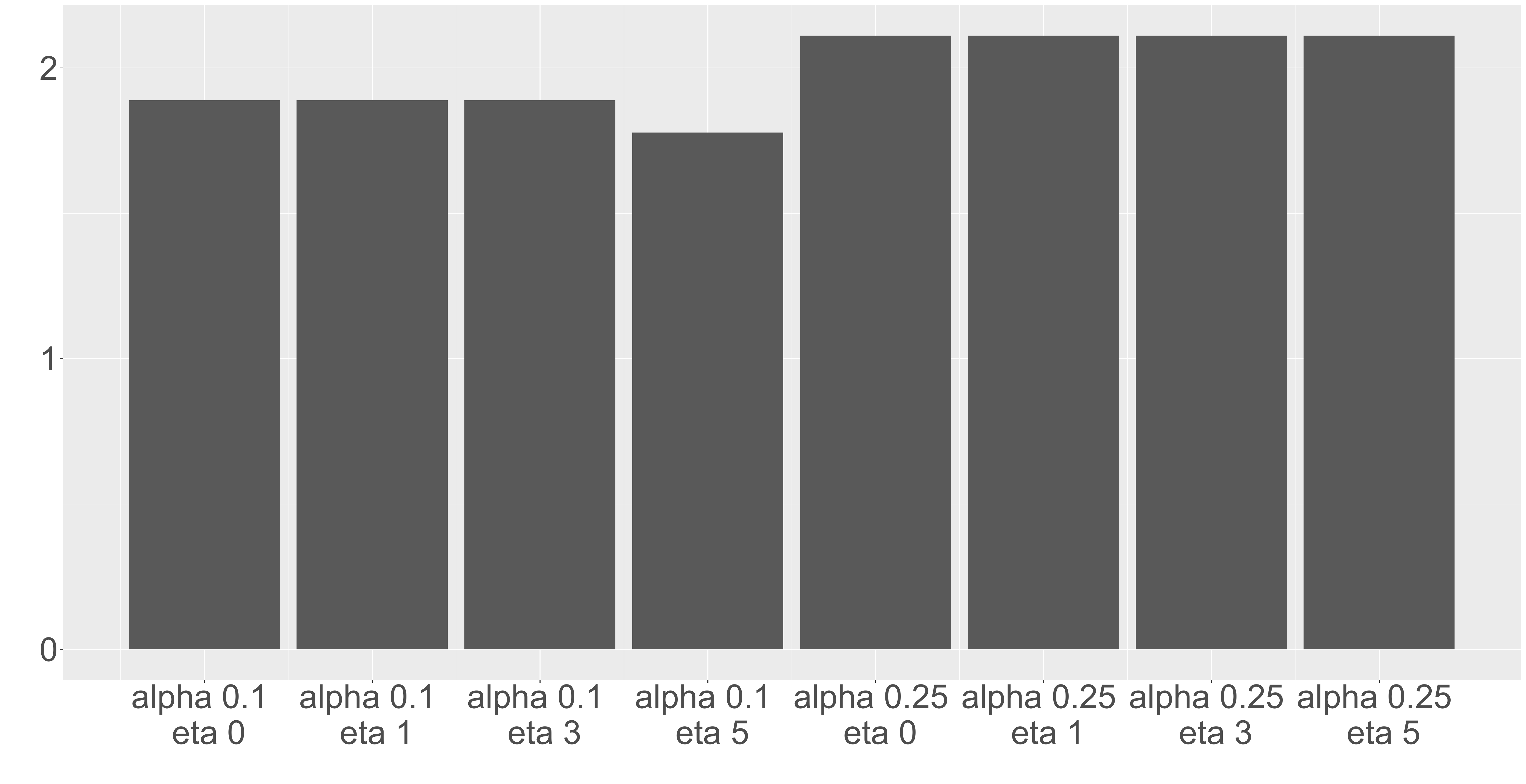}  & \includegraphics[width=.33\textwidth]{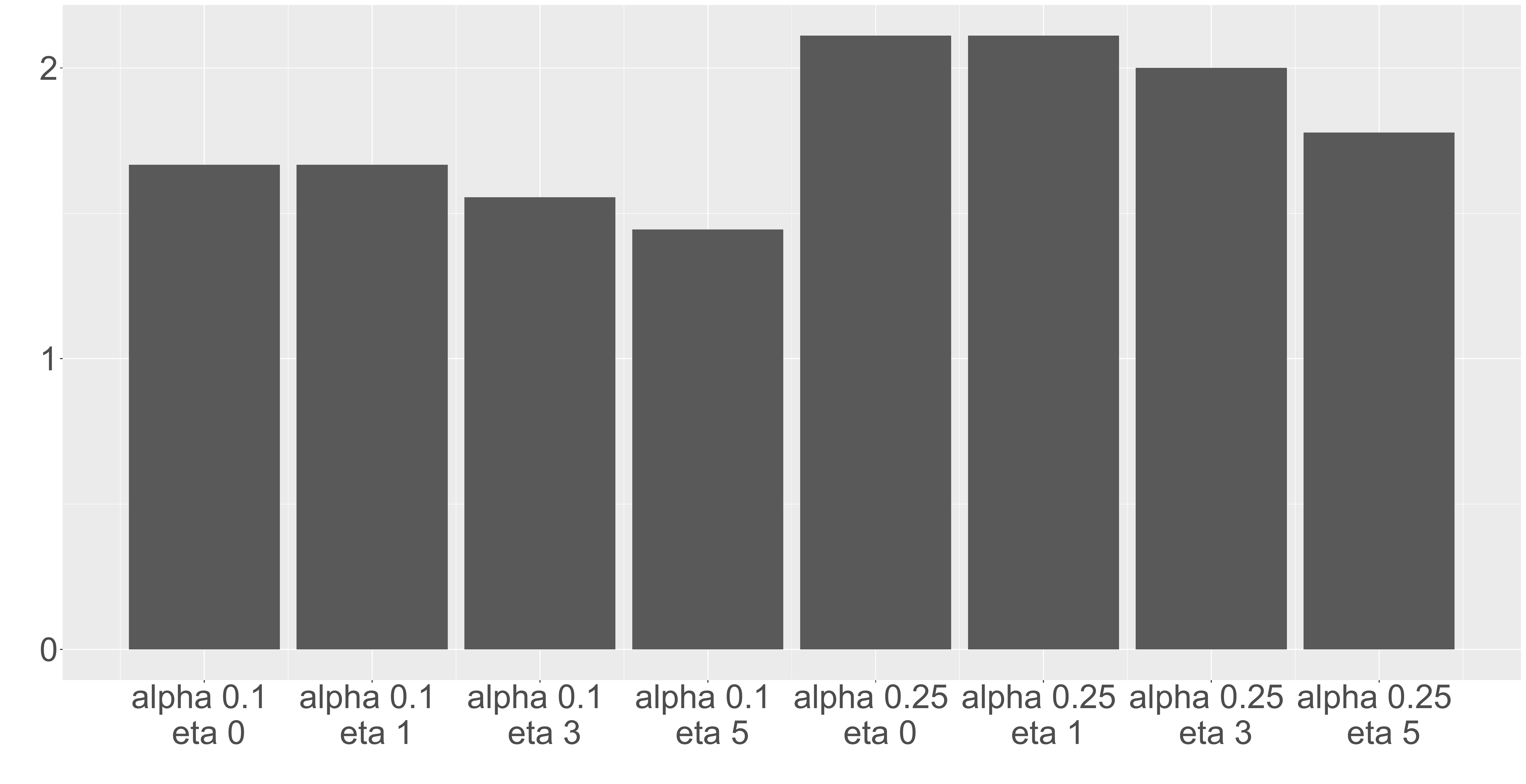}\vspace{-2mm}\\
{\footnotesize (a) Erd\"os-R\'eny social graph} & {\footnotesize (b) Power-law \mbox{$\delta\!=\!0.05$} social graph} & {\footnotesize (c) Power-law \mbox{$\delta\!=\!0.1$} social graph}
\end{tabular}
\captionsetup{justification=justified,singlelinecheck=off,font={stretch=0.7}}
\vspace{-4mm}\caption{\footnotesize Synthetic data: execution time of the proposed \us\ method (milliseconds), by varying the $\alpha$ and $\eta$ parameters and the social graph ($|\A| = 1000$, noise~$5\%$, BIC regularizator).\vspace{-4mm}} \label{fig:time_simulations}
\end{figure*}

We generate synthetic data for $100$ independent runs. 
In each run we generate social graphs of $n=100$ nodes, of the form of either Erd\"os-R\'eny  (with density $\delta$ uniformly sampled from the range $[0.05, 0.1]$), or power-law (with density $\delta$ being either $0.05$ or $0.1$).  
We partition each randomly generated social network into $k=10$ clusters of $\mbox{card}_{min} = 8$, $\mbox{card}_{max} = 12$ and $\mbox{card}_{overlap} = 10$ nodes ($1\%$). 
For each cluster we generate causal \DAGs\ of density $\delta_{G_{cause}}$ uniformly sampled from the range $[0.35,0.5]$, and $p_{min},p_{max}$ uniformly sampled from $(0,1)$. 

We consider both a noise-free model and a noisy model.
As for the former, we assume a perfect regularity in the actions of directly connected users: all the actions of a child user follow those of one of its parents, thus letting us  constrain the induced distribution of the generating \DAG as follows. 
With $u$ being a user in the network and $P(u)$ being all the users with arcs pointing to $u$ (i.e., $u$'s parent set), the probability of observing any action from $u$ is $0$ if none of $u$'s parents had performed an action before (i.e., $P(u)$ are independently influencing $u$).
As far as the noisy model, we consider probabilities $e_+, e_- \in (0,1)$ of false-positive and false-negative observations in each item, respectively, i.e., the probability of adding or removing from $\A$ a certain triple $\langle v, \phi, t \rangle$, regardless of how it was sampled from the underlying causal process.
Note that these two sources of noise aim at modeling imperfect regularities in the causal phenomena in terms of either false positive/negative observations or noise in the orderings among the events.
We ultimately generate our noisy datasets with a probability of a noisy entry of either $5\%$ or $10\%$ per entry, with any noisy entry having uniform probability of being either a false positive or a false negative.

Given the above settings, we sample observation sets at different sizes, i.e., $|\A|\!=\!500$, $|\A|\!=\!1000$, and $|\A|\!=\!5000$.
The observations are sampled across the $k=10$ groups in such a way that the size of each group is guaranteed to be within $\frac{|\A|}{k} \times (1 \pm [0,0.5])$, e.g., for size $|\A|= 500$, every cluster will have between $25$ and $75$ traces.

\begin{table}[t]
\vspace{-3mm}
\centering
\footnotesize
\captionsetup{justification=justified,singlelinecheck=off,font={stretch=0.7}}
\vspace{3mm}
\caption{{ \footnotesize Synthetic data: performance of the proposed \us\ method vs. the baseline, by varying the $\alpha$ and $\eta$ parameters, on the Erd\"os-R\'eny social graph ($|\A| = 1000$, noise $5\%$, BIC regularizator).}}
\vspace{-4mm}
\begin{tabular}{@{}c@{ }@{  }r@{ }|cccc|cccc@{}}
& \multicolumn{1}{c}{} & \multicolumn{4}{c}{$\alpha = 0.1$} & \multicolumn{4}{|c}{$\alpha = 0.25$}\\
\cline{3-10}
& \multicolumn{1}{c}{} & $\eta\!=\!0$ & $\eta\!=\!1$ & $\eta\!=\!3$ & $\eta\!=\!5$ & $\eta\!=\!0$ & $\eta\!=\!1$ & $\eta\!=\!3$ & $\eta\!=\!5$\\
\hline
\multirow{2}{*}{$accuracy$} & \us & 0.932 &	0.932 &	0.933 &	0.934 &	0.933 &	0.933 &	0.933 &	0.934\\
& \baseline &  0.765 &	0.767 &	0.786 &	0.793 &	0.773 &	0.774 &	0.791 &	0.799\\
\hline
\multicolumn{2}{c|}{NMI} & 0.67 &	0.669 &	0.674 &	0.677 &	0.671 &	0.671 &	0.675 &	0.679\\

\hline
\end{tabular}
\label{tab:exp1-erdos}
\end{table}

\begin{table}[t]
\centering
\footnotesize
\captionsetup{justification=justified,singlelinecheck=off,font={stretch=0.7}}
\vspace{3mm}
\caption{{ \footnotesize Synthetic data: performance of the proposed \us\ method vs. the baseline, by varying the $\alpha$ and $\eta$ parameters, on the power-law $\delta\!=\!0.05$ social graph ($|\A| = 1000$, noise $5\%$, BIC regularizator).}}
\vspace{-4mm}
\begin{tabular}{@{}c@{ }@{  }r@{ }|cccc|cccc@{}}
& \multicolumn{1}{c}{} & \multicolumn{4}{c}{$\alpha = 0.1$} & \multicolumn{4}{|c}{$\alpha = 0.25$}\\
\cline{3-10}
& \multicolumn{1}{c}{} & $\eta\!=\!0$ & $\eta\!=\!1$ & $\eta\!=\!3$ & $\eta\!=\!5$ & $\eta\!=\!0$ & $\eta\!=\!1$ & $\eta\!=\!3$ & $\eta\!=\!5$\\
\hline
\multirow{2}{*}{$accuracy$} & \us & 0.979 & 0.979 & 0.979 & 0.98 & 0.979 & 0.978 & 0.979 & 0.98\\
& \baseline & 0.938 & 0.938 & 0.942 & 0.944 & 0.938 & 0.938 & 0.941 & 0.945\\
\hline
\multicolumn{2}{c|}{NMI} & 0.563	& 0.563	& 0.563	& 0.563	& 0.563	& 0.563	& 0.563 &	0.563\\
\hline
\end{tabular}
\label{tab:exp1-powerlaw0.05}
\end{table}

\begin{table}[t]
\centering
\footnotesize
\captionsetup{justification=justified,singlelinecheck=off,font={stretch=0.7}}
\vspace{3mm}
\caption{{ \footnotesize Synthetic data: performance of the proposed \us\ method vs. the baseline, by varying the $\alpha$ and $\eta$ parameters, on the power-law $\delta\!=\!0.1$ social graph ($|\A| = 1000$, noise $5\%$, BIC regularizator).}}
\vspace{-4mm}
\begin{tabular}{@{}c@{ }@{  }r@{ }|cccc|cccc@{}}
& \multicolumn{1}{c}{} & \multicolumn{4}{c}{$\alpha = 0.1$} & \multicolumn{4}{|c}{$\alpha = 0.25$}\\
\cline{3-10}
& \multicolumn{1}{c}{} & $\eta\!=\!0$ & $\eta\!=\!1$ & $\eta\!=\!3$ & $\eta\!=\!5$ & $\eta\!=\!0$ & $\eta\!=\!1$ & $\eta\!=\!3$ & $\eta\!=\!5$\\
\hline
\multirow{2}{*}{$accuracy$} & \us & 0.966	& 0.965	& 0.967	& 0.968	& 0.965	& 0.965	& 0.967	& 0.968\\
& \baseline & 0.882	& 0.882	& 0.888	& 0.892	& 0.883	& 0.882	& 0.887	& 0.893\\
\hline
\multicolumn{2}{c|}{NMI} & 0.63	& 0.63	& 0.631	& 0.631	& 0.63	& 0.63	& 0.63	& 0.63\\
\hline
\end{tabular}
\label{tab:exp1-powerlaw0.1}
\end{table}

\begin{table*}[t]
\vspace{-6mm}
\centering
\footnotesize
\captionsetup{justification=justified,singlelinecheck=off,font={stretch=0.7}}
\vspace{3mm}
\caption{{ \footnotesize Synthetic data: performance of the proposed \us\ method vs. the baseline, by varying the size $|\A|$ of input observations ($\alpha\!=\!0.1$, $\eta\!=\!1$, noise $5\%$, BIC regularizator).}}
\vspace{-6mm}
\begin{tabular}{@{}c@{ }@{  }r@{ }|ccc|ccc|ccc@{}}
& \multicolumn{1}{c}{} & \multicolumn{3}{c}{Erd\"os-R\'eny} & \multicolumn{3}{|c}{Power-law $\delta\!=\!0.05$} & \multicolumn{3}{|c}{Power-law $\delta\!=\!0.1$}\\
\cline{3-11}
& \multicolumn{1}{c}{} & $|\A|\!=\!500$ & $|\A|\!=\!1000$ & $|\A|\!=\!5000$ & $|\A|\!=\!500$ & $|\A|\!=\!1000$ & $|\A|\!=\!5000$ & $|\A|\!=\!500$ & $|\A|\!=\!1000$ & $|\A|\!=\!5000$\\
\hline
\multirow{2}{*}{$accuracy$} & \us  &  0.939  &  0.932  &  0.909  &  0.983  &  0.979  &  0.963  &  0.974  &  0.965  &  0.938\\
& \baseline &  0.815  &  0.767  &  0.585  &  0.95  &  0.938  &  0.913  &  0.904  &  0.882  &  0.835\\
\hline
\multicolumn{2}{c|}{NMI}   &  0.669  &  0.662  &  0.662  &  0.563  &  0.567  &  0.571  &  0.630  &  0.636  &  0.642\\
\hline
\end{tabular}
\label{tab:exp2}
\end{table*}

\begin{table*}[t]
\vspace{-2mm}
\centering
\footnotesize
\captionsetup{justification=justified,singlelinecheck=off,font={stretch=0.7}}
\vspace{3mm}
\caption{{ \footnotesize Synthetic data: performance of the proposed \us\ method vs. the baseline, by varying the noise level ($\alpha\!=\!0.1$, $\eta\!=\!1$, $|\A|\!=\!1000$, BIC regularizator).}}
\vspace{-4mm}
\begin{tabular}{@{}r@{ }|c@{ }@{ }cc@{ }@{ }cc@{ }@{ }cc@{ }@{ }c|c@{ }@{ }cc@{ }@{ }cc@{ }@{ }cc@{ }@{ }c@{}}
 \multicolumn{1}{c}{} & \multicolumn{8}{c}{$\alpha = 0.1$} & \multicolumn{8}{|c}{$\alpha = 0.25$}\\
\cline{2-17}
 \multicolumn{1}{c}{} &  \multicolumn{2}{c}{$\eta\!=\!0$} & \multicolumn{2}{c}{$\eta\!=\!1$} & \multicolumn{2}{c}{$\eta\!=\!3$} & \multicolumn{2}{c|}{$\eta\!=\!5$} & \multicolumn{2}{c}{$\eta\!=\!0$} & \multicolumn{2}{c}{$\eta\!=\!1$} & \multicolumn{2}{c}{$\eta\!=\!3$} & \multicolumn{2}{c}{$\eta\!=\!5$}\\
\cline{2-17}
 \multicolumn{1}{c}{} & BIC & AIC & BIC & AIC & BIC & AIC & BIC & AIC & BIC & AIC & BIC & AIC & BIC & AIC & BIC & AIC\\
\hline
 $accuracy$ & 0.979  &  0.971  &  0.979  &  0.971  &  0.979  &  0.972  &  0.98  &  0.973  &  0.979  &  0.971  &  0.978  &  0.971  &  0.979  &  0.972  &  0.98  &  0.973\\
\hline
\end{tabular}
\label{tab:exp3}
\end{table*}

\begin{table*}[t]
\vspace{-2mm}
\centering
\footnotesize
\captionsetup{justification=justified,singlelinecheck=off,font={stretch=0.7}}
\vspace{3mm}
\caption{{ \footnotesize Synthetic data: performance of the proposed \us\ method with varying the regularizator, i.e., BIC vs. AIC ($|\A| = 1000$, noise $5\%$, power-law $\delta\!=\!0.05$ social graph).}}
\vspace{-6mm}
\begin{tabular}{@{}c@{ }@{  }r@{ }|ccc|ccc|ccc@{}}
& \multicolumn{1}{c}{} & \multicolumn{3}{c}{Erd\"os-R\'eny} & \multicolumn{3}{|c}{Power-law $\delta\!=\!0.05$} & \multicolumn{3}{|c}{Power-law $\delta\!=\!0.1$}\\
\cline{3-11}
& \multicolumn{1}{c}{} & no noise & noise $5\%$ & noise $10\%$ & no noise & noise $5\%$ & noise $10\%$ & no noise & noise $5\%$ & noise $10\%$\\
\hline
\multirow{2}{*}{$accuracy$} & \us  & 0.936  &  0.932  &  0.93  &  0.98  &  0.979  &  0.978  &  0.967  &  0.965  &  0.964 \\
& \baseline &  0.839  &  0.767  &  0.713  &  0.941  &  0.938  &  0.936  &  0.887  &  0.882  &  0.878 \\
\hline
\multicolumn{2}{c|}{NMI}   &  0.686  &  0.669  &  0.661  &  0.563  &  0.563  &  0.563  &  0.63  &  0.63  &  0.63\\
\hline
\end{tabular}
\vspace{-0mm}
\label{tab:exp4}
\end{table*}

\spara{Assessment criteria.}
We use the following metrics:
\begin{itemize}
\item {\em Causal topology.} \ We assess how well the causal structure ultimately inferred by our method reflects the generative model.
For this assessment we resort to the traditional Hamming distance.
Specifically, we first build the union graph of the causal \DAGs of each cluster to obtain the graph of all the true causal claims. 
We do this for both the generative \DAGs (i.e., the ground-truth ones) and the inferred ones (i.e., the ones yielded by the proposed method or the baseline). 
Then, we compute the Hamming distance from the ground-truth and the inferred structures, i.e., we count the minimum number of substitutions required to remove any inconsistency from the output topologies, when compared to the ground-truth ones. 
We ultimately report the performance in terms of $accuracy = \frac{(TP + TN)}{(TP + TN + FP + FN)}$, with $TP$ and $FP$ being the arcs recognized as true and false positives, respectively, and $TN$ and $FN$ being the arcs recognized as true and false negatives, respectively.
\item {\em Partitioning.} \ We also estimate how well our partitioning step can effectively group users involved in the same causal process. 
To this end, we measure the similarity between the clusters identified by our method and the ground-truth clusters by means of the well-established Normalized Mutual Information (NMI) measure~\cite{danon2005comparing}.
\end{itemize}

\spara{Results.}
We compare the performance of the proposed two-step \textsf{Probabilistic Social Causation} method (for short, \us) against a baseline \baseline\ that only performs the first one of the two steps of the \us\ algorithm (i.e., only the grouping step as described in Algorithm~\ref{alg:proposed-sampling}), and reconstruct a \DAG $G_D(\mathcal{D})$ from the prima-facie graph $G(\mathcal{D})$ of every group $\mathcal{D} \in \mathbf{D}^*$ outputted by that step, without learning the minimal causal topology.
All results are averaged over the 100 data-generation runs performed, and, unless otherwise specified, they refer to $K\!=\!100$ (for the proposed \us\ method).

First, we evaluate the execution time of the proposed \us\ method by varying the algorithm parameters (i.e., $\alpha$ and $\eta$), and the type of social graph underlying the generated data. 
The results of this experiment are reported in Figure~\ref{fig:time_simulations}. 
As depicted in the figure, the running time of the proposed method is in the order of a few milliseconds.
Also, the different values of $\alpha$ and $\eta$, as well as the form of the social graph, do not seem to have a significant impact on the execution time.

Shifting the attention to effectiveness, Tables~\ref{tab:exp1-erdos}--\ref{tab:exp1-powerlaw0.1} report the performance of the proposed \us\ method vs. the \baseline\ baseline, by varying the algorithm parameters ($\alpha$ and $\eta$), on the various social graphs.
In all cases the $accuracy$ of the proposed \us\ is evidently higher than the one of the baseline.
The performance is rather independent of the algorithm parameters or the social graph.
In terms of NMI (which is the same for both \us\ and \baseline, as it concerns the first step of the proposed method that is common to \us\ and \baseline),  the performance is in the range $[0.56, 0.68]$, which is a fair result considering the difficulty of the subtask of recognizing the exact ground-truth cluster structure.

Table~\ref{tab:exp2} reports on the performance by varying the number of sampled observations.
As expected, the trends follow the common intuition: the performance of both \us\ and \baseline\ decreases as the number of observations increases.
However, a major result to remark here is that the advantage of the proposed \us\ over \baseline\ gets higher with the increasing observations.
This attests to the effectiveness of our method even for large observation-set sizes.

The last experiments we focus on are on the impact of  the regularizator and the noise level on the performance of \us.
The results of these experiments are shown in Table~\ref{tab:exp3}~and~\ref{tab:exp4}, respectively.
As far as the former, BIC is recognized as slightly more accurate than AIC.
As for the noise level, we observe only a slight decrease of the performance of the proposed \us\ as the noise level increases, which attests to the high robustness of the proposed method.

\subsection{Real data} \label{sec:real_data}

We also experiment with three real-world datasets, whose main characteristics are summarized in Table~\ref{tab:dataset}. 

\begin{table}[h]
\vspace{-2mm}
\centering
\footnotesize
\captionsetup{justification=justified,singlelinecheck=off,font={stretch=0.7}}
\caption{{ \footnotesize Characteristics of real data: number of observations ($|\A|$); number of propagations/\DAGs\ ($|\D|$); nodes ($|V|$) and arcs ($|A|$) of the social graph $G$; min, max, and avg number of nodes in a \DAG\ of $\D$ ($n_{min}$, $n_{max}$, $n_{avg}$);  min, max, and avg number of arcs in a \DAG\ of $\D$ ($m_{min}$, $m_{max}$, $m_{avg}$).}\vspace{-3mm}}
\begin{tabular}{@{}c@{}@{  }|r@{}@{  }r@{}@{  }r@{}@{  }r@{}@{  }|r@{}@{  }r@{}@{  }r@{}@{  }r@{}@{  }r@{}@{  }r@{}}
 \multicolumn{1}{c}{} & $|\A|$ & $|\D|$ & $|V|$ & $|A|$ & $n_{min}$ & $n_{max}$ & $n_{avg}$  & $m_{min}$ & $m_{max}$  & $m_{avg}$\\
\hline
\lastfm  &  1\,208\,640  &  51\,495  &   1\,372 &  14\,708 & 6 & 472 & 24 & 5 &  2\,704 & 39\\
\twitter  &  580\,141  &  8\,888  &  28\,185  &  1\,636\,451  & 12 &  13\,547  & 66 & 11 &  240\,153  & 347\\
\flixster  &  6\,529\,012  &  11\,659  &  29\,357  &  425\,228  & 14 &  16\,129  & 561 & 13 &  85\,165  &  1\,561 \\
\hline
\end{tabular}
\vspace{-2mm}
\label{tab:dataset}
\end{table}

\lastfm\ (\texttt{www.last.fm}).
Lastfm is a music website, where users listen to their favorite tracks and communicate with each other. The dataset was created starting from the \textit{HetRec 2011 Workshop} dataset available at \url{www.grouplens.org/node/462}, and enriching it by crawling. The graph $G$ corresponds to the friendship network of the service. The entities in $\E$ are the songs listened to by the users. An observation $\langle u,\phi, t \rangle \in \A$ means that the first time that the user $u$ listens to the song $\phi$ occurs at time $t$.

\twitter\ (\texttt{twitter.com}).
We obtained the dataset by crawling the public timeline of the popular online microblogging service.
The nodes of the graph $G$ are the Twitter users, while each arc $(u,v)$ expresses the fact that $v$ is a follower of $u$. The entities in $\E$ correspond to URLs, while an observation $\langle u,\phi, t \rangle \in \A$ means that the user $u$  (re-)tweets (for the first time) the URL $\phi$ at time $t$.

\flixster\ (\texttt{www.flixster.com}).
Flixster is a social movie site where people can meet each other based on tastes in movies. The graph $G$ corresponds to the social network underlying the site. The entities in $\E$ are movies, and an observation $\langle u, \phi, t\rangle$ is included in $\A$ when the user $u$ rates for the first time the movie $\phi$ with the rating happening at time $t$.
%

%

\begin{figure*}[t]
\centering
\vspace{-4mm}
\begin{tabular}{ccc}
\includegraphics[width=.32\textwidth]{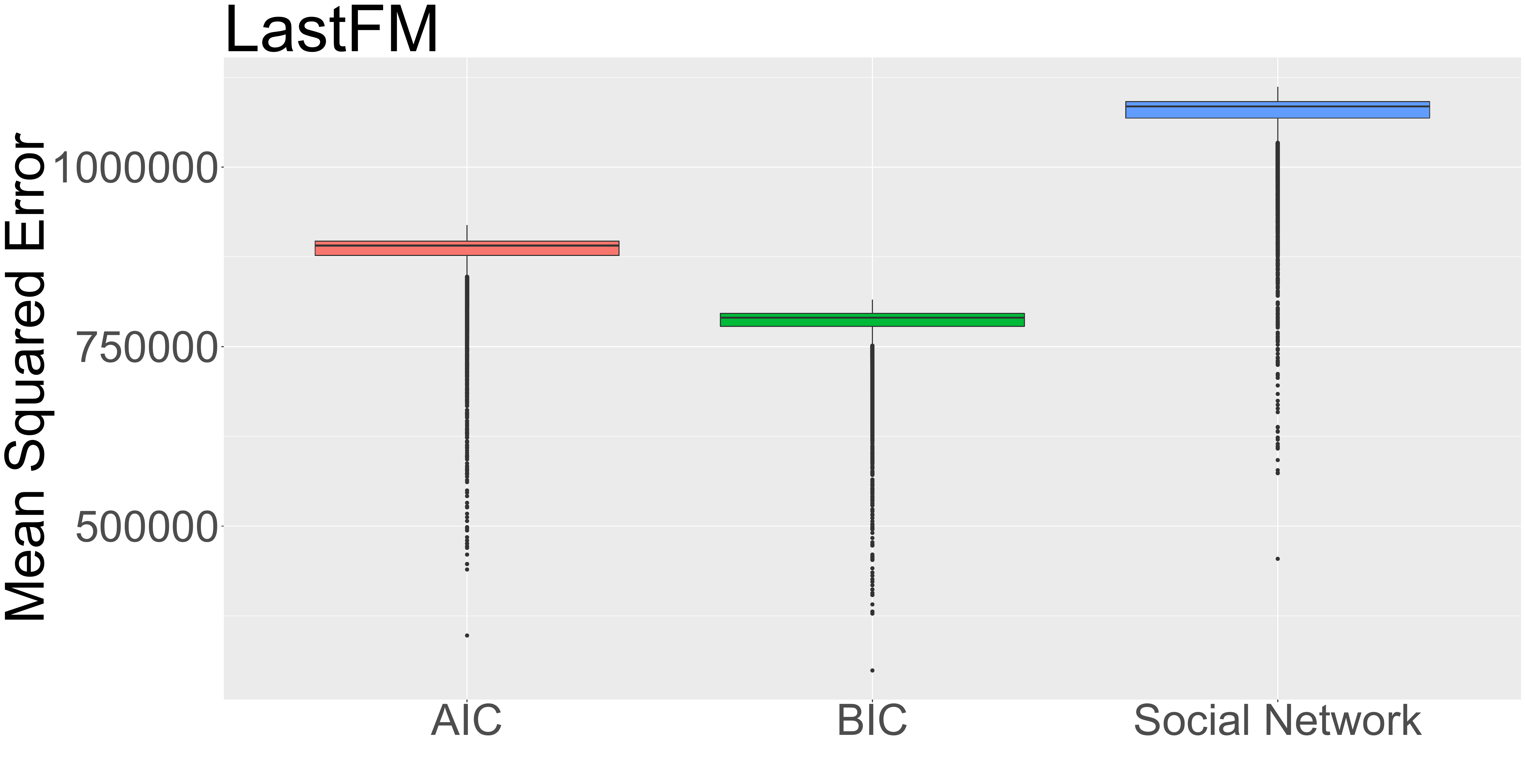} & \includegraphics[width=.32\textwidth]{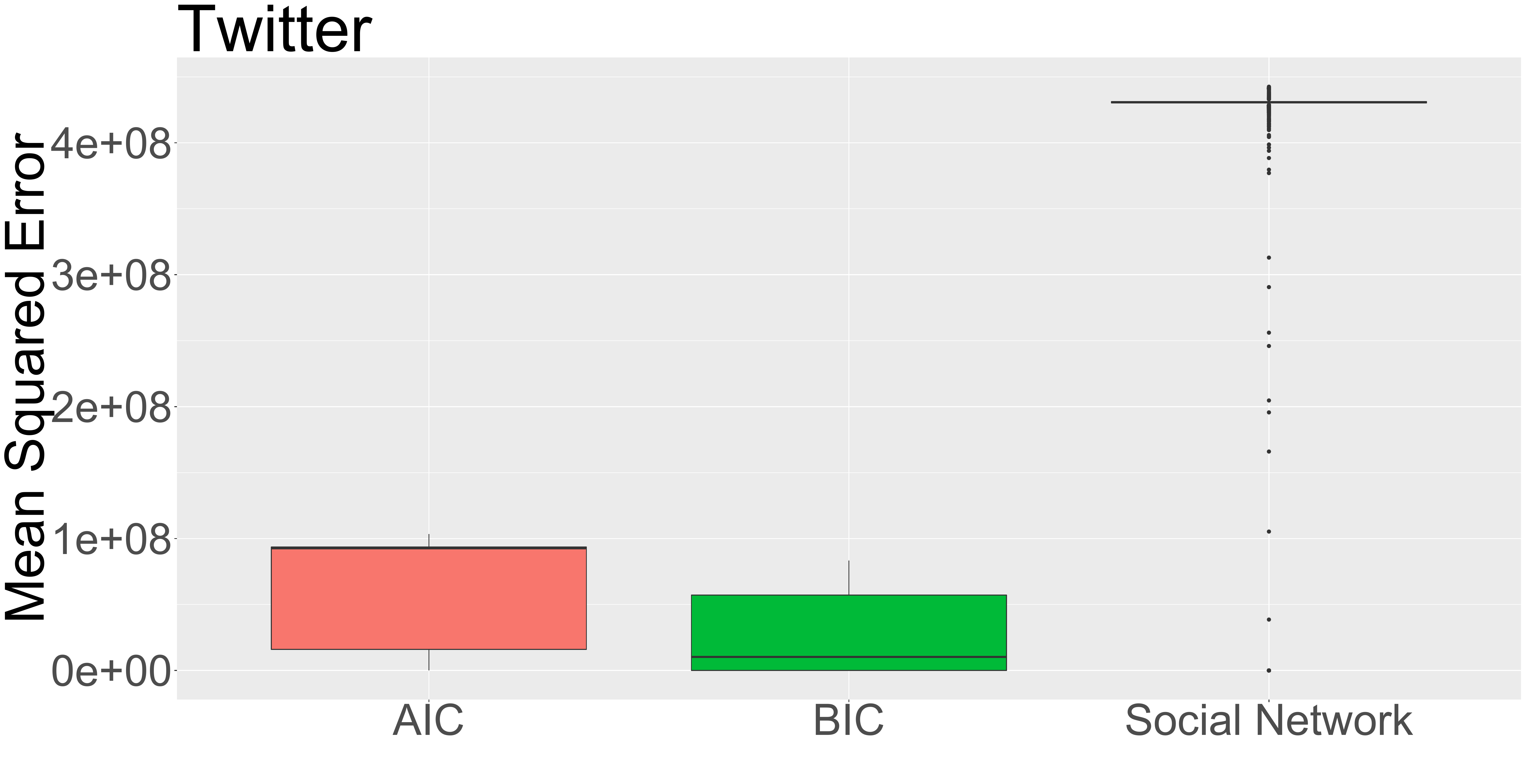}  & \includegraphics[width=.32\textwidth]{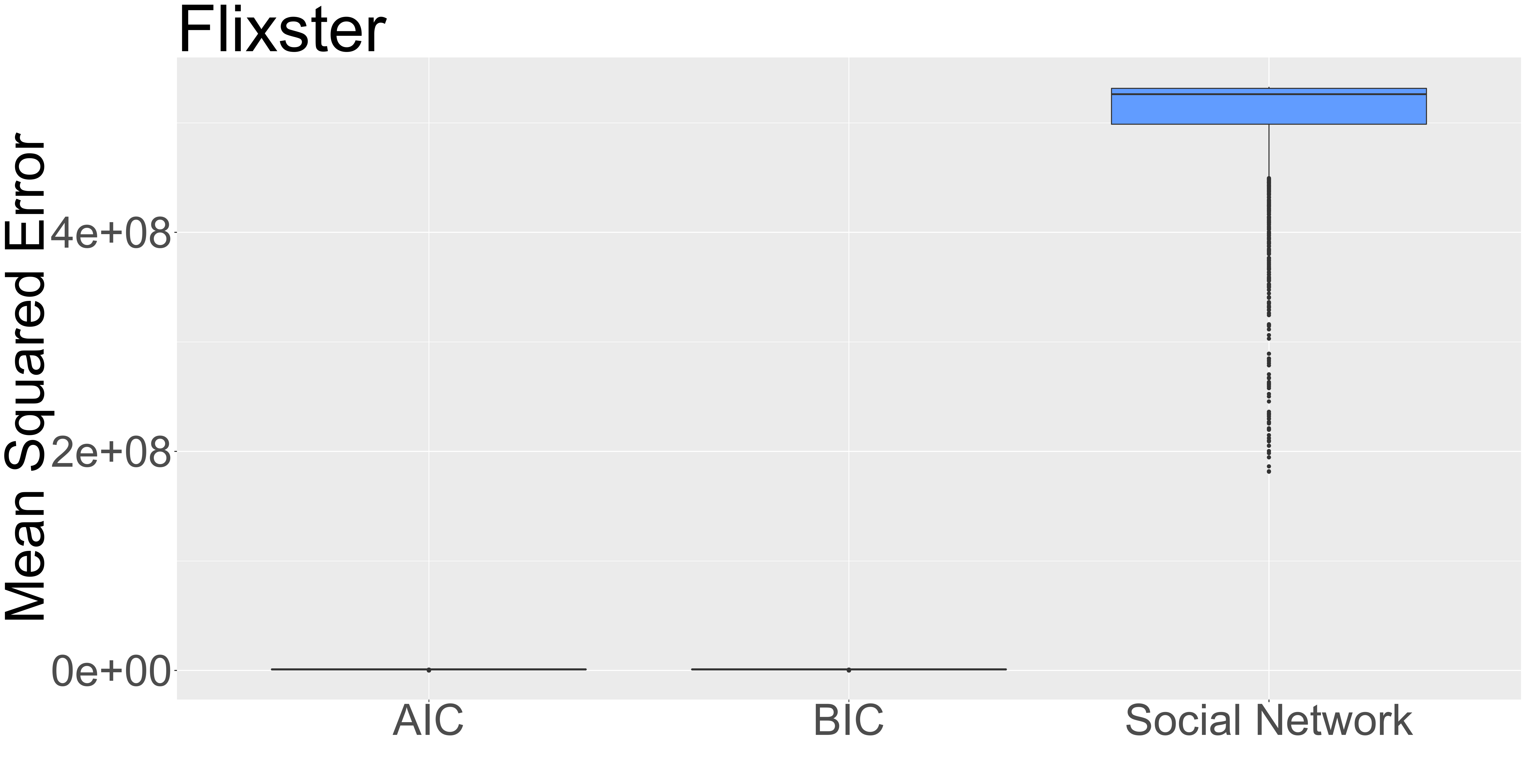}\vspace{-2mm}
\end{tabular}
\captionsetup{justification=justified,singlelinecheck=off,font={stretch=0.7}}
\vspace{-4mm}\caption{\footnotesize Real-data: spread-prediction performance of the proposed \us\ method (equipped with BIC or AIC regularizator) vs. a baseline that considers the whole social graph ($\alpha=0.2$, $\eta=5$).\vspace{-0mm}} \label{fig:mse}
\vspace{-3mm}
\end{figure*}

\smallskip

As real data comes with no ground-truth, here we resort to the well-established \emph{spread-prediction} task to assess the effectiveness of our method. This task aims at predicting the expected number of nodes that eventually get activated due to an information-propagation process initiated in some nodes~\cite{goyal2011data}. Specifically, in our experiments we consider the well-established propagation model defined by Goyal~{\em et~al.}~\cite{goyal2011data}, which takes as input a graph (representing relationships among some users) and a database of propagations (representing actions taken by those users), and learns a model that is able to predict the influence spread.
In our context we (randomly) split our input propagations into training set and test set (70\% vs. 30\%), and use the training set to learn the Goyal~{\em et~al.}'s model.
As a graph, we consider both the ultimate causal structure computed by our \us\  method, and the whole input social graph.
The latter constitutes a baseline to evaluate the effectiveness of  \us: the rationale is that providing the Goyal~{\em et~al.}'s model with a graph corresponding to the causal structure recognized by our method, instead of the whole social network, is expected to improve the performance of the spread-prediction task, as it will not burden the model with noisy relationships.
Once the Goyal~{\em et~al.}'s model has been learned, we use it to predict spread, and measure the accuracy of the experiment in terms of \emph{mean squared error} (MSE) between the predicted spreads and the real ones exhibited by the test set.

The results of this experiment are shown in Figure~\ref{fig:mse}, where for our \us\ method we set $\alpha = 0.20$, $\eta = 5$, and $K = 500$, while for both \us\ and the baseline we perform spread prediction by running $10,000$ Monte Carlo simulations (as suggested in~\cite{goyal2011data}).
The figure shows that, on all datasets, the error is consistently smaller when our method is used compared to when the original social network is given as input to the spread-prediction algorithm.
This finding hints at a promising capability achievable by the proposed method in extracting the real causal relations within the social network.

%% file: conclusions.tex

In this paper we tackle the problem of deriving causal \DAGs that are well-representative of the social-influence dynamics underlying an input database of propagation traces.
We devise a principled two-step methodology that is based on  Suppes' probabilistic-causation theory.
The first step of the methodology aims at partitioning the input set of propagations, mainly to get rid of the Simpson's paradox, while the second step derives the ultimate minimal causal topology via constrained MLE.
Experiments on synthetic data attest to the high accuracy of the proposed method in detecting ground-truth causal structures, while experiments on real data show that our method performs well in a task of spread prediction.

%% file: main.bbl
\begin{thebibliography}{10}

\bibitem{agarwal2002algorithmic}
P.~K. Agarwal, L.~J. Guibas, H.~Edelsbrunner, J.~Erickson, M.~Isard,
  S.~Har-Peled, J.~Hershberger, C.~Jensen, L.~Kavraki, P.~Koehl, et~al.
\newblock Algorithmic issues in modeling motion.
\newblock {\em ACM Computing Surveys (CSUR)}, 34(4):550--572, 2002.

\bibitem{CharuBook}
C.~C. Aggarwal and J.~Han.
\newblock {\em Frequent Pattern Mining}.
\newblock Springer, 2014.

\bibitem{apriori}
R.~Agrawal and R.~Srikant.
\newblock Fast algorithms for mining association rules in large databases.
\newblock In {\em VLDB}, pages 487--499, 1994.

\bibitem{akaike1974new}
H.~Akaike.
\newblock A new look at the statistical model identification.
\newblock {\em IEEE Trans. Aut. Contr.}, 19(6):716--723, 1974.

\bibitem{aldrich1995correlations}
J.~Aldrich.
\newblock Correlations genuine and spurious in pearson and yule.
\newblock {\em Statistical Science}, pages 364--376, 1995.

\bibitem{aris08}
A.~Anagnostopoulos, R.~Kumar, and M.~Mahdian.
\newblock Influence and correlation in social networks.
\newblock In {\em KDD}, pages 7--15, 2008.

\bibitem{aral2009distinguishing}
S.~Aral, L.~Muchnik, and A.~Sundararajan.
\newblock {Distinguishing influence-based contagion from homophily-driven
  diffusion in dynamic networks}.
\newblock {\em PNAS}, 106(51):21544--21549, 2009.

\bibitem{bakshy2011}
E.~Bakshy, J.~M. Hofman, W.~A. Mason, and D.~J. Watts.
\newblock Everyone's an influencer: quantifying influence on twitter.
\newblock In {\em WSDM}, pages 65--74, 2011.

\bibitem{surveybonchi11}
F.~Bonchi.
\newblock Influence propagation in social networks: A data mining perspective.
\newblock {\em IEEE Intelligent Informatics Bulletin, Vol.12 No.1: 8-16,
  December 2011}.

\bibitem{bonchi2017exposing}
F.~Bonchi, S.~Hajian, B.~Mishra, and D.~Ramazzotti.
\newblock Exposing the probabilistic causal structure of discrimination.
\newblock {\em IJDSA}, 3(1):1--21, 2017.

\bibitem{caravagna2016algorithmic}
G.~Caravagna, A.~Graudenzi, D.~Ramazzotti, R.~Sanz-Pamplona, L.~De~Sano,
  G.~Mauri, V.~Moreno, M.~Antoniotti, and B.~Mishra.
\newblock Algorithmic methods to infer the evolutionary trajectories in cancer
  progression.
\newblock {\em PNAS}, 113(28):E4025--E4034, 2016.

\bibitem{CastilloMP11}
C.~Castillo, M.~Mendoza, and B.~Poblete.
\newblock Information credibility on {T}witter.
\newblock In {\em WWW}, pages 675--684, 2011.

\bibitem{ChaMG09}
M.~Cha, A.~Mislove, and P.~K. Gummadi.
\newblock A measurement-driven analysis of information propagation in the
  {F}lickr social network.
\newblock In {\em WWW}, pages 721--730, 2009.

\bibitem{learning_NP_hard}
D.~M. Chickering, D.~Heckerman, and C.~Meek.
\newblock Large-sample learning of bayesian networks is np-hard.
\newblock {\em JMLR}, 5:1287--1330, 2004.

\bibitem{Cook}
S.~Cook.
\newblock The {C}omplexity of {T}heorem-proving {P}rocedures.
\newblock In {\em STOC}, pages 151--158, 1971.

\bibitem{Cormode2010}
G.~Cormode, H.~Karloff, and A.~Wirth.
\newblock Set cover algorithms for very large datasets.
\newblock In {\em CIKM}, pages 479--488, 2010.

\bibitem{CrandallCHKS08}
D.~J. Crandall, D.~Cosley, D.~P. Huttenlocher, J.~M. Kleinberg, and S.~Suri.
\newblock Feedback effects between similarity and social influence in online
  communities.
\newblock In {\em KDD}, pages 160--168, 2008.

\bibitem{danon2005comparing}
L.~Danon, A.~Diaz-Guilera, J.~Duch, and A.~Arenas.
\newblock Comparing community structure identification.
\newblock {\em JSTAT}, 2005(09):P09008, 2005.

\bibitem{domingos01}
P.~Domingos and M.~Richardson.
\newblock Mining the network value of customers.
\newblock In {\em KDD}, pages 57--66, 2001.

\bibitem{FondN10}
T.~L. Fond and J.~Neville.
\newblock Randomization tests for distinguishing social influence and homophily
  effects.
\newblock In {\em WWW}, pages 601--610, 2010.

\bibitem{gao2017efficient}
G.~Gao, B.~Mishra, and D.~Ramazzotti.
\newblock Efficient simulation of financial stress testing scenarios with
  suppes-bayes causal networks.
\newblock {\em Procedia Computer Science}, 108:272--284, 2017.

\bibitem{Golbeck06}
J.~Golbeck and J.~Hendler.
\newblock Inferring binary trust relationships in web-based social networks.
\newblock {\em ACM Trans. Internet Technol.}, 6(4):497--529, 2006.

\bibitem{goyal2011data}
A.~Goyal, F.~Bonchi, and L.~V. Lakshmanan.
\newblock A data-based approach to social influence maximization.
\newblock {\em PVLDB}, 5(1):73--84, 2011.

\bibitem{amit2010}
A.~Goyal, F.~Bonchi, and L.~V.~S. Lakshmanan.
\newblock Learning influence probabilities in social networks.
\newblock In {\em WSDM}, pages 241--250, 2010.

\bibitem{Guha04}
R.~Guha, R.~Kumar, P.~Raghavan, and A.~Tomkins.
\newblock Propagation of trust and distrust.
\newblock In {\em WWW}, 2004.

\bibitem{GupteSLMI}
M.~Gupte, P.~Shankar, J.~Li, S.~Muthukrishnan, and L.~Iftode.
\newblock Finding hierarchy in directed online social networks.
\newblock In {\em WWW}, pages 557--566, 2011.

\bibitem{AiZaki2009}
M.~A. Hasan and M.~J. Zaki.
\newblock Musk: Uniform sampling of k maximal patterns.
\newblock In {\em SDM}, pages 650--661, 2009.

\bibitem{hill}
S.~Hill, F.~Provost, and C.~Volinsky.
\newblock Network-based marketing: Identifying likely adopters via consumer
  networks.
\newblock {\em Statistical Science}, 21(2):256--276, 2006.

\bibitem{hitchcock1997probabilistic}
C.~Hitchcock.
\newblock Probabilistic causation.
\newblock In E.~N. Zalta, editor, {\em The Stanford Encyclopedia of
  Philosophy}. Winter 2012 edition, 2012.

\bibitem{memerank}
D.~Ienco, F.~Bonchi, and C.~Castillo.
\newblock The meme ranking problem: Maximizing microblogging virality.
\newblock In {\em SIASP IEEE ICDM Work.}, 2010.

\bibitem{kempe03}
D.~Kempe, J.~M. Kleinberg, and {\'E}.~Tardos.
\newblock Maximizing the spread of influence through a social network.
\newblock In {\em KDD}, pages 137--146, 2003.

\bibitem{kleinberg2010algorithmic}
S.~Kleinberg.
\newblock {\em An algorithmic enquiry concerning causality}.
\newblock PhD thesis, New York University, 2010.

\bibitem{koller2009probabilistic}
D.~Koller and N.~Friedman.
\newblock {\em Probabilistic graphical models: principles and techniques}.
\newblock MIT press, 2009.

\bibitem{KutzkovBBG13}
K.~Kutzkov, A.~Bifet, F.~Bonchi, and A.~Gionis.
\newblock {STRIP:} stream learning of influence probabilities.
\newblock In {\em KDD}, pages 275--283, 2013.

\bibitem{LeskovecAH07}
J.~Leskovec, L.~A. Adamic, and B.~A. Huberman.
\newblock The dynamics of viral marketing.
\newblock {\em TWEB}, 1(1), 2007.

\bibitem{LeskovecSK06}
J.~Leskovec, A.~Singh, and J.~M. Kleinberg.
\newblock Patterns of influence in a recommendation network.
\newblock In {\em PAKDD}, pages 380--389, 2006.

\bibitem{caprese_causation}
L.~O. Loohuis, G.~Caravagna, A.~Graudenzi, D.~Ramazzotti, G.~Mauri,
  M.~Antoniotti, and B.~Mishra.
\newblock Inferring tree causal models of cancer progression with probability
  raising.
\newblock {\em PloS ONE}, 9(10):e108358, 2014.

\bibitem{MoensGoethals2013}
S.~Moens and B.~Goethals.
\newblock {Randomly Sampling Maximal Itemsets}.
\newblock In {\em KDD IDEA Work.}, pages 79--86, 2013.

\bibitem{newman2003mixing}
M.~E. Newman.
\newblock Mixing patterns in networks.
\newblock {\em Physical Review E}, 67(2):026126, 2003.

\bibitem{pearl2009causality}
J.~Pearl.
\newblock {\em Causality}.
\newblock Cambridge university press, 2009.

\bibitem{capri_causation}
D.~Ramazzotti, G.~Caravagna, L.~Olde~Loohuis, A.~Graudenzi, I.~Korsunsky,
  G.~Mauri, M.~Antoniotti, and B.~Mishra.
\newblock {CAPRI}: efficient inference of cancer progression models from
  cross-sectional data.
\newblock {\em Bioinformatics}, 31(18):3016--3026, 2015.

\bibitem{RomeroMK11}
D.~M. Romero, B.~Meeder, and J.~M. Kleinberg.
\newblock Differences in the mechanics of information diffusion across topics:
  idioms, political hashtags, and complex contagion on {Twitter}.
\newblock In {\em WWW}, pages 695--704, 2011.

\bibitem{Saito08}
K.~Saito, R.~Nakano, and M.~Kimura.
\newblock Prediction of information diffusion probabilities for independent
  cascade model.
\newblock In {\em KES}, pages 67--75, 2008.

\bibitem{NectaRSS}
J.~J. Samper, P.~A. Castillo, L.~Araujo, and J.~J.~M. Guerv{\'o}s.
\newblock Nectarss, an rss feed ranking system that implicitly learns user
  preferences.
\newblock {\em CoRR}, abs/cs/0610019, 2006.

\bibitem{bic_1978}
G.~Schwarz.
\newblock Estimating the dimension of a model.
\newblock {\em Annals of Statistics}, 6(2):461--464, 1978.

\bibitem{shalizi2011homophily}
C.~R. Shalizi and A.~C. Thomas.
\newblock Homophily and contagion are generically confounded in observational
  social network studies.
\newblock {\em Sociological methods \& research}, 40(2):211--239, 2011.

\bibitem{sharma2016distinguishing}
A.~Sharma and D.~Cosley.
\newblock Distinguishing between personal preferences and social influence in
  online activity feeds.
\newblock In {\em CSCW}, pages 1091--1103, 2016.

\bibitem{song07}
X.~Song, Y.~Chi, K.~Hino, and B.~L. Tseng.
\newblock Information flow modeling based on diffusion rate for prediction and
  ranking.
\newblock In {\em WWW}, pages 191--200, 2007.

\bibitem{Song06}
X.~Song, B.~L. Tseng, C.-Y. Lin, and M.-T. Sun.
\newblock Personalized recommendation driven by information flow.
\newblock In {\em SIGIR}, pages 509--516, 2006.

\bibitem{suppes_prima_facie}
P.~Suppes.
\newblock {\em A Probabilistic Theory of Causality}.
\newblock North-Holland Publishing Company, 1970.

\bibitem{Mohsen08}
M.~Taherian, M.~Amini, and R.~Jalili.
\newblock Trust inference in web-based social networks using resistive
  networks.
\newblock In {\em ICIW}, pages 233--238, 2008.

\bibitem{Tatti14}
N.~Tatti.
\newblock Faster way to agony: Discovering hierarchies in directed graphs.
\newblock In {\em ECML PKDD}, pages 163--178, 2014.

\bibitem{Tatti17}
N.~Tatti.
\newblock Tiers for peers: a practical algorithm for discovering hierarchy in
  weighted networks.
\newblock {\em DAMI}, 31(3):702--738, 2017.

\bibitem{wagner1982simpson}
C.~H. Wagner.
\newblock Simpson's paradox in real life.
\newblock {\em The American Statistician}, 36(1):46--48, 1982.

\bibitem{weng2010}
J.~Weng, E.-P. Lim, J.~Jiang, and Q.~He.
\newblock {TwitterRank}: finding topic-sensitive influential twitterers.
\newblock In {\em WSDM}, pages 261--270, 2010.

\bibitem{Ziegler05}
C.-N. Ziegler and G.~Lausen.
\newblock Propagation models for trust and distrust in social networks.
\newblock {\em Information Systems Frontiers}, 7(4-5):337--358, 2005.

\end{thebibliography}
